\documentclass[10pt,journal]{IEEEtran}

\usepackage{amssymb}
\usepackage{amsmath}
\usepackage{amsthm}
\usepackage{mathtools}
\usepackage{graphicx}
\usepackage{xfrac}
\usepackage{color}
\usepackage[linesnumbered,lined]{algorithm2e}
\usepackage{tikz}
\usepackage[USenglish]{babel}

\newcommand{\comment}[1]{}

 \newtheorem{theorem}{Theorem}
 \newtheorem{lemma}{Lemma}
\newtheorem{prop}{Proposition}

\title{Explicit back-off rates for achieving target throughputs in CSMA/CA networks}

\author{B. Van Houdt\\
Dept. Mathematics and Computer Science\\
University of Antwerp, Belgium}
\date{}

\begin{document}

\maketitle

\begin{abstract}
CSMA/CA networks have often been analyzed using a stylized model
that is fully characterized by a vector of back-off rates
and a conflict graph. Further, for any achievable throughput vector $\vec \theta$
the existence of a unique vector $\vec \nu(\vec \theta)$ of back-off rates
that achieves this throughput vector was proven. Although this unique
vector can in principle be computed iteratively, the required time complexity 
 grows exponentially in the network size, making this only feasible for small networks.

In this paper, we present an explicit formula for the unique vector of back-off rates  $\vec \nu(\vec \theta)$ needed
to achieve any achievable throughput vector $\vec \theta$ provided that the network has a 
{\it chordal} conflict graph. This class of networks contains
a number of special cases of interest such as (inhomogeneous) line networks and networks with an 
acyclic conflict graph.
Moreover, these back-off rates are such that the back-off rate of a node only depends
on its own target throughput and the target throughput of its neighbors and can be
determined in a {\it distributed} manner.

We further indicate that back-off rates of this form cannot be obtained in general for networks
with non-chordal conflict graphs. For general conflict graphs we nevertheless show how to adapt the
back-off rates when a node is added to the network when its interfering nodes
form a clique in the conflict graph. 
Finally, we introduce a distributed chordal approximation algorithm
for general conflict graphs which is shown (using numerical examples)
to be more accurate than the Bethe approximation.
\end{abstract}


\section{Introduction}
Consider a carrier-sense multiple access (CSMA) network with collision avoidance (CA) where a group of nodes shares
a single communication channel. Several nodes in such a network can transmit packets simultaneously, but
whenever two (or more) transmissions interfere with each other a {\it collision} takes place, which typically results
in a failed transmission. As such nodes in a CSMA network try to avoid collisions by {\it sensing} the channel before
packet transmission. If the channel is sensed busy, a node postpones its transmission attempt for some time. 

An often studied model for CSMA/CA networks is the so-called {\it ideal} model 
\cite{boorstyn1,durvy2,durvy1,jiang2,jiang1,vandeven3,vandeven5,vandeven1,wang5,yun1}. The ideal CSMA/CA
model considers a network with a fixed set of $n$ nodes\footnote{The nodes in the conflict graph
are often regarded as being the links in the real network.} and is fully characterized by a fixed conflict graph $G$ and a fixed
vector of back-off rates $\vec \nu = (\nu_1,\ldots,\nu_n)$. The conflict graph $G$ identifies the pairs of nodes that interfere with each other,
while the vector $(\nu_1,\ldots,\nu_n)$ determines the mean time that the nodes have 
to sense the channel idle before they are allowed to start a transmission. 

One of the key assumptions of the ideal CSMA/CA model is that sensing is instantaneous, which implies that collisions 
cannot occur (as the probability that two nodes start transmitting at exactly the same time is zero). Another
important assumption is that each of the $n$ nodes always has packets ready for transmission, that is,
the network is assumed to be saturated. Further perfect sensing and packet transmission is assumed.
While these assumptions clearly do not hold in practice, this model was shown to predict the throughput 
of real CSMA/CA based ad-hoc networks in a surprisingly accurate manner in some cases \cite{wang5} 
(especially when the maximum contention window is large). Further, networks with saturated users can be
used to study the stability of networks with unsaturated users that behave in a greedy manner, i.e., that
transmit a dummy packet whenever their buffer is empty.

While the product form solution for the steady state probabilities of the ideal CSMA/CA model has been
established long ago \cite{boorstyn1} and the set $\Gamma$ of achievable throughput vectors $\vec \theta = (\theta_1,\ldots,\theta_n)$
has been identified in \cite{jiang1}, very few explicit results are available on how to set the back-off rates 
$\vec \nu$ to achieve a given vector $\vec \theta \in \Gamma$ (where $\Gamma$ clearly depends on the conflict graph $G$). 
In \cite{vandeven1} an explicit formula was presented to achieve fairness in a line network where each
node interferes with its $\beta$ left and right nodes. More recently, by relying on some existing results in statistical physics, 
explicit formulas for the back-off vector needed to achieve a given throughput vector were presented in case the
conflict graph is a tree \cite{yun1}. 
The existence of a unique vector of back-off rates for each achievable throughput vector was proven in \cite{vandeven5},
where iterative algorithms to compute this unique vector were discussed. While these iterative algorithms are guaranteed to
converge, computation times can easily become prohibitive as they grow exponentially in the network size.

The main objective of this paper is to identify the set of conflict graphs $G$ for which simple explicit expressions can be
obtained for the vector of back-off rates $\vec \nu$ that achieves a given throughput vector $\vec \theta \in \Gamma$.
We show that an explicit expression can be obtained for any {\it chordal} conflict graph, thereby generalizing existing
results for line networks and networks that have a tree as a conflict graph. In fact we present two
explicit expressions for the required back-off rates: one based on a clique\footnote{A clique in a graph is a subset of nodes such that
its induced subgraph is complete.} tree and one that relies on a perfect elimination ordering of $G$.
Somewhat surprisingly these explicit expressions are such that the back-off rate $\nu_i$ of node $i$ only
depends on its own target throughput $\theta_i$ and the target throughput of its neighbors in the conflict graph $G$.

We also explore whether such explicit expressions can be obtained for non-chordal conflict graphs and show
that even in the simplest case (a ring network of size $4$), this is no longer the case. We do however prove 
a property on how to adapt the back-off rates when adding a node in a particular manner to a
general conflict graph and this property also explains why explicit expression can be obtained for chordal conflict graphs.
Further we also look into the possibility of defining a chordal approximation for networks with non-chordal conflict
graphs and propose a distributed chordal approximation that is more accurate than the Bethe approximation (for the numerical
experiments conducted). 

The paper is structured as follows.  In Section \ref{sec:model} we give a detailed description of the ideal CSMA/CA model
under consideration. In Section \ref{sec:special} we briefly discuss the case of acyclic conflict graphs
and (inhomogeneous) line networks. Section \ref{sec:chordal} contains the main results of the paper and presents explicit expressions
for the back-off rates when the conflict graph is chordal. In Section \ref{sec:approx} we explore  
the possibility of defining chordal approximations for general conflict graphs. 
Some results for non-chordal conflict graphs are presented in Section \ref{sec:beyond}.
Finally, related work is discussed in Section \ref{sec:related} and conclusions are drawn in Section \ref{sec:conclusion}.

\section{Model description}\label{sec:model}
Consider a network consisting of $n$ nodes that is fully characterized
by a vector of back-off rates $(\nu_1,\ldots,\nu_n)$ and an {\it undirected} conflict graph $G=(V(G),E(G))$, with $V(G)=\{1,\ldots,n\}$. 
A node is either active or inactive
at any point in time. The conflict graph $G$ specifies which pairs of nodes cannot be simultaneously active,
that is, nodes $i$ and $j$ cannot be active simultaneously if and only if $(i,j) \in E(G)$.
When a node becomes active, it remains active for some time before
becoming inactive again. An inactive node can only become active if none of its neighbors in $G$ are active.
When a node becomes inactive it starts a back-off period. As soon as one of the neighbors of an inactive node in $G$ becomes active, the 
back-off period of the inactive node is frozen and resumes when all of its neighbors are inactive again. A node becomes active when its  back-off
period ends. 

If the duration of the active period is exponential (with mean $1$) and the back-off period is exponentially
distributed with mean $1/\nu_i$ for node $i$, it is well-known \cite{boorstyn1} that this network evolves as a reversible Markov chain on the state space 
\[\Omega = \{(z_1,\ldots,z_n) \in \{0,1\}^n | z_i z_j = 0 \mbox{ if } (i,j) \in E(G)\},   \]
where node $i$ is active in state  $(z_1,\ldots,z_n)$ if and only if $z_i = 1$.
The steady state probabilities $\pi(\vec z)$, with $\vec z = (z_1,\ldots,z_n)$ of this Markov chain are given by
\begin{align}\label{eq:stst}
\pi(\vec z) = \frac{1}{Z_n} \prod_{i=1}^n \nu_i^{z_i}, 
\end{align}  
where 
\[Z_n = \sum_{\vec z \in \Omega}  \prod_{i=1}^n \nu_i^{z_i},\] 
is the normalizing constant.
Clearly, the throughput $\theta_i$ of node $i$ can be expressed as
\[\theta_i = \sum_{\vec z \in \Omega, z_i = 1} \pi(\vec z), \]
for $i=1,\ldots,n$. In \cite{jiang1} the set of achievable throughput vectors $\vec \theta = (\theta_1,\ldots,\theta_n)$
was shown to equal 
\begin{align}\label{eq:Gamma} 
\Gamma =  \left\{ \sum_{\vec z \in \Omega} \xi(\vec z) \vec z \middle| \sum_{\vec z \in \Omega} \xi(\vec z) =1, \xi(\vec z) > 0 \mbox{ for } \vec z \in \Omega \right\}.
\end{align}
In other words, a throughput vector $\vec \theta$ is achievable if and only if it belongs to the interior of the
convex hull of the set $\Omega$. Further, for each achievable vector $\vec \theta \in \Gamma$ there exists a unique vector $\vec \nu = (\nu_1,\ldots,\nu_n)$
of back-off rates that achieves $\vec \theta$ \cite{vandeven5}.

If the duration of an active period follows a phase-type distribution\footnote{The class of phase-type distributions is dense in the class of positive valued distributions.}, one obtains a somewhat more complicated Markov chain. However given that the mean of this distribution is $1$, one can show that all the nodes
have the same throughput as in the exponential case \cite{boorstyn1}. The same is true if we additionally
replace the exponential back-off period by one with a phase type distribution (with the same mean), 
irrespective of whether a node freezes its the back-off period when a neighbor becomes active \cite{vandeven3}.
If the back-off period is not frozen and expires when a neighbor is active, the node simply starts a new back-off period.
In short, the model considered in this paper is not restricted to exponential back-off periods and activity periods.

\section{Some special cases}\label{sec:special}
Before presenting the general result for chordal conflict graphs in the next section we discuss
some special cases first, that is, we first discuss acyclic conflict graphs, homogeneous and inhomogeneous line networks. 
One way to see that these are special cases is to note that the conflict graph of a homogeneous line network belongs to the class of 
$K$-trees\footnote{Any $K$-tree can be obtained by starting with the complete graph with $K+1$ nodes and
subsequently adding nodes one at a time by connecting a node to a clique of $K$ nodes.} (with $K=\beta$), while the one of an inhomogeneous line network belongs to the class of interval graphs\footnote{For any interval graph one can associate a set of consecutive integers to each node such that two
nodes are adjacent if and only if their corresponding sets of consecutive integers intersect.}.
Further, acyclic graphs are $1$-trees and any $K$-tree or interval graph is chordal.   
 
\subsection{Acyclic conflict graphs}\label{sec:tree}
In this section we introduce an explicit formula to set the back-off rates to achieve a given
achievable target throughput vector $\vec \theta = (\theta_1,\ldots,\theta_n)$ in a network consisting of $n$ nodes
provided that its conflict graph $G=(V(G),E(G))$ is acyclic, that is, it is a tree (as we assume interference is symmetric,
meaning $G$ is undirected).

Let $\mathcal{N}_i$ be the set of neighbors of $i$ in $G$ and define $\nu_i^{tree}(\vec \theta) $ as 
\begin{align}\label{eq:nui_tree}
\nu_i^{tree}(\vec \theta) = 
\frac{  \theta_i (1-\theta_i)^{|\mathcal{N}_i|-1}}
{ \prod_{j \in \mathcal{N}_i} (1-(\theta_i+\theta_j))},
\end{align}
for $i=1,\ldots,n$.
It is worth noting that the rate $\nu_i^{tree}(\vec \theta)$ is fully determined by the target throughput $\theta_i$ and
those of
its neighbors, meaning to set its back-off rate it suffices for a node to know the target throughput of the interfering
nodes. A direct proof of the following theorem is given in Appendix \ref{apx:proof_tree}:

\begin{theorem}\label{TH:TREE} 
Let $\vec \theta = (\theta_1,\ldots,\theta_n)$ be a positive vector with $\hat T=\max_{(k,j) \in E} (\theta_k+\theta_j) < 1$.
The throughput of node $i$, for $i=1,\ldots,n$, in a network with an acyclic conflict graph $G$ matches 
$\theta_i$ if and only if the back-off rates are set according to \eqref{eq:nui_tree}.
\end{theorem}

\noindent {\bf Remarks.}

1) The back-off rates given by \eqref{eq:nui_tree} were independently introduced in \cite[Equation 6]{yun1} as the Bethe approximation
for the CSMA stability problem. In \cite{yun1} the problem of finding the required back-off rates to achieve a given
throughput vector $\vec \theta$ is restated as a problem of minimizing the Gibbs free energy (GFE).
The minimizer of the GFE is subsequently approximated by the minimizer of the
Bethe free energy (BFE) and the rates in \eqref{eq:nui_tree} are shown to be a minimizer of the BFE. 
When the conflict graph is acyclic the BFE and GFE coincide, as such 
the result in Theorem \ref{TH:TREE} also follows from \cite[Section III.C]{yun1}. 

2) When $\theta_1 = \ldots = \theta_n = \gamma < 1/2$,  $\nu_i^{tree}(\vec \theta)$ simplifies to
\[ \nu_i^{tree}(\vec \theta) = \gamma \frac{(1-\gamma)^{|\mathcal{N}_i|-1}}
{(1-2\gamma)^{|\mathcal{N}_i|}},\]
which is a generalization of \cite[p383]{kelly2}, where a system is considered that has
an acyclic conflict graph where all the non-leaf nodes have $r$ neighbors.

3) The condition $\hat T < 1$ is clearly a necessary condition for $\vec \theta$ to be achievable.
It is also sufficient as $(\nu_1^{tree}(\vec \theta),\ldots,$ $\nu_n^{tree}(\vec \theta))$ is a well-defined back-off
vector that achieves these throughputs. In other words the set of achievable throughput vectors $\Gamma$, defined in
\eqref{eq:Gamma}, is given by
\[\Gamma = \{(\theta_1,\ldots,\theta_n) | \theta_i + \theta_j < 1, \mbox{ for } (i,j) \in E \}.\]

\subsection{Line networks}\label{sec:line}
The aim of this section is to show that in a line network consisting of $n$ nodes with an interference range of $\beta$,
any given
achievable target throughput vector $\vec \theta = (\theta_1,\ldots,\theta_n)$ is achieved if and only if the back-off rate
of node $i$ is set equal to 
\begin{align}\label{eq:nui_line}
\nu_i^{line}(\vec \theta) = \theta_i \
\frac{ \displaystyle \prod_{j=\max(i,\beta+1)}^{\min(i+\beta,n)-1} (1-(\theta_{j-\beta+1}+\ldots+\theta_j))}
{\displaystyle \prod_{j=\max(i,\beta+1)}^{\min(i+\beta,n)} (1-(\theta_{j-\beta}+\ldots+\theta_j))},
\end{align}
for $i=1,\ldots,n$.

The rate $\nu_i^{line}(\vec \theta)$ is fully determined by $\theta_{i-\beta},
\ldots,\theta_{i+\beta}$, that is, by the target throughput of its interfering nodes (as in the acyclic case).
Further note that $\beta = 1$ is the only case for which the corresponding conflict graph is acyclic as for $\beta > 1$ 
nodes $1$ to $\beta+1$ form a clique in the conflict graph. 
A direct proof of the next theorem is presented in Appendix \ref{apx:proof_line}:

\begin{theorem}\label{TH:LINE} 
Let $\vec \theta = (\theta_1,\ldots,\theta_n)$ be a positive vector with $T=\max_{i=1}^{n-\beta} (\theta_i+\ldots+\theta_{i+\beta}) < 1$.
The throughput of node $i$, for $i=1,\ldots,n$, in a line network with interference range $\beta$ 
matches $\theta_i$ if and only if the back-off rates are set according to \eqref{eq:nui_line}.
\end{theorem}

\noindent {\bf Remarks.} 

1) When $\theta_1 = \ldots = \theta_n = \gamma < 1/(\beta+1)$,  $\nu_i^{line}(\vec \theta)$ simplifies to
\[ \nu_i^{line}(\vec \theta) = \gamma \frac{(1-\gamma \beta)^{h_i-1}}
{(1-\gamma (\beta+1))^{h_i}},\] 
with $h_i = \min(i+\beta,n)-\max(i,\beta+1)+1$. This result is equivalent to 
\cite[Proposition 3]{vandeven5} which is a restatement of the main result in \cite{vandeven1}.

2) The condition $T < 1$ is a necessary and sufficient condition for the vector $\vec \theta$ to be an achievable
throughput vector, that is, $\Gamma$ defined in \eqref{eq:Gamma} can be written as
\[\Gamma =\{(\theta_1,\ldots,\theta_n)| \max_{i=1}^{n-\beta} (\theta_i + \ldots + \theta_{i+\beta}) < 1\}.\]

3) While Theorem \ref{TH:LINE} applies only to line networks, it can be readily used to obtain 
results for networks obtained by replacing node $i$ by a clique consisting of $k_i$ nodes, for $i=1,\ldots,n$, where 
we  assume that the nodes part of clique $i$ interfere with all the nodes part of the next and previous $\beta$
cliques. Indeed, for the nodes part of clique $i$ it makes no difference whether the
other cliques consist of a single or multiple nodes as long as the sum of their back-off rates remains the same.
Hence, if we replace node $i$ by $k_i$ nodes, then all {\it nodes} achieve a 
throughput of $\gamma < 1/K$, with $K=\max_{i=1}^{n-\beta} (k_i+\ldots+k_{i+\beta})$, if and only if
the back-off rate $\nu_i^{cline}$ of a node part of clique $i$ is given by 
\[ \nu_i^{cline} = \gamma
\frac{\prod_{j=\max(i,\beta+1)}^{\min(i+\beta,n)-1} (1-\gamma(k_{j-\beta+1}+\ldots+k_j))}
{\prod_{j=\max(i,\beta+1)}^{\min(i+\beta,n)} (1-\gamma(k_{j-\beta}+\ldots+k_j))},\]
due to Theorem \ref{TH:LINE}.

\subsection{Inhomogeneous line networks}
The line networks considered so far can be generalized by relaxing the requirement that each node
interferes with the previous (and next) $\beta$ nodes. More specifically, let $\vec \beta = (\beta_1,\ldots,\beta_{n+1})$ 
be a vector with $\beta_1 =\beta_{n+1}= 0$, $\beta_{i} > 0$ integer and $\beta_{i+1} \leq \beta_i+1$, for $i=2,\ldots,n$.
Define the conflict graph $G(\vec \beta) = (V(\vec \beta),E(\vec \beta))$ such that
\begin{align*}
V(\vec \beta) &= \{1,\ldots,n\},\\
E(\vec \beta) &= \{ (i,j) | i=1,\ldots,n; j=i-\beta_i,\ldots,i-1\}.
\end{align*}
In other words, the (undirected) conflict graph is connected and node $i$ interferes with the previous $\beta_i$ nodes (and potentially
with some of the next nodes). The requirement that $\beta_{i+1} \leq \beta_i+1$ indicates that node $i+1$ cannot interfere
with some node $j<i$ if node $i$ did not interfere with node $j$. 
Hence, this requirement guarantees that the nodes $\{i-\beta_i, \ldots, i\}$ form a clique of size $\beta_i+1$ in $G(\vec \beta)$.
Remark that the line networks considered in the previous section correspond to the case 
\[\vec \beta = (0,1,\ldots,\beta-1,\underbrace{\beta, \ldots, \beta}_{n-\beta},0).\]
Based on the vector $\vec \beta$ we construct an ordered set of $m = |\{ i | \beta_{i} \geq \beta_{i+1}\}|$ cliques
denoted as $K_1,\ldots,K_m$.
These $m$ cliques are the maximal cliques of $G(\vec \beta)$ in the sense that for any $K_j$ there is no
larger clique in $G(\vec \beta)$ that contains $K_j$. Indeed, if $\beta_{i} \geq \beta_{i+1}$, then
the nodes $\{i-\beta_i,\ldots,i\}$ form a maximal clique. We order these $m$ maximal cliques such that
 $\min_{s \in K_i} s < \min_{s\in K_{i+1}} s$, for $i=1,\ldots,m-1$.

Consider the example in Figure \ref{fig:iline} where $n=9$ and $\vec \beta = (0,1,1,2,1,2,3,2,2,0)$.
In this example the ordered list of $m=5$ maximal cliques in $G(\vec \beta)$  is
given by $K_1 = \{1,2\},$ $K_2 = \{2,3,4\},$ $K_3 = \{4,5,6,7\},$ $K_4 = \{6,7,8\},$ and $K_5 = \{7,8,9\}$.
Finally, let $f(i)$ and $l(i)$ denote the index of the first and last maximal clique in which node $i$ appears, respectively.
In the example above we have $f(2) = 1, l(2)=2, f(5)=l(5)=3$, $f(7)=3$ and $l(7)=5$. 

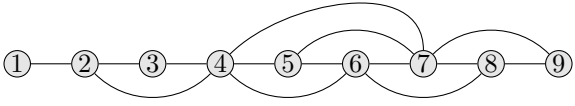
\begin{figure}[t]
\begin{center}
\begin{tikzpicture}[scale=0.9]
\tikzstyle{every node}=[circle, draw, fill=black!10,
                        inner sep=0pt, minimum width=10pt]

\node (1) at (1,0) []{$1$};
\node (2) at (2,0) []{$2$};
\node (3) at (3,0) []{$3$};
\node (4) at (4,0) []{$4$};
\node (5) at (5,0) []{$5$};
\node (6) at (6,0) []{$6$};
\node (7) at (7,0) []{$7$};
\node (8) at (8,0) []{$8$};
\node (9) at (9,0) []{$9$};

\draw [] (1) -- (2);
\draw [] (2) -- (3);
\draw [] (3) -- (4);
\draw [] (4) -- (5);
\draw [] (5) -- (6);
\draw [] (6) -- (7);
\draw [] (7) -- (8);
\draw [] (8) -- (9);
\draw (2) to[out=-45,in=-135] (4);
\draw (4) to[out=-45,in=-135] (6);
\draw (6) to[out=-45,in=-135] (8);
\draw (5) to[out=45,in=135] (7);
\draw (7) to[out=45,in=135] (9);
\draw (4) to[out=45,in=90] (7);

\end{tikzpicture}
\end{center}
\caption{Conflict graph $G(\vec \beta)$ of inhomogeneous line network with $\vec \beta = (0,1,1,2,1,2,3,2,2,0)$.}
\label{fig:iline}
\end{figure}

We are now in a position to define the rates $\nu_i^{iline}(\vec \theta)$ for any achievable 
throughput vector $\vec \theta = (\theta_1,\ldots,\theta_n)$ as
follows:
\begin{align}\label{eq:nui_iline}
\nu_i^{iline}(\vec \theta) = \theta_i \
\frac{ \displaystyle \prod_{j=f(i)}^{l(i)-1} \left(1-\sum_{s\in K_j \cap K_{j+1}} \theta_s\right)}
{\displaystyle \prod_{j=f(i)}^{l(i)}  \left(1-\sum_{s\in K_j} \theta_s\right)},
\end{align}
for $i=1,\ldots,n$. In order to set its back-off rate equal to $\nu_i^{iline}(\vec \theta)$ node $i$ needs to
know the target throughput of its interfering nodes as well as the exact composition of the maximal cliques it
belongs to.

\begin{theorem}\label{th:achieve_iline} 
Consider an inhomogeneous line network that is characterized by the vector $\vec \beta$
and let $\vec \theta = (\theta_1,\ldots,\theta_n)$ be a positive vector with $T=\max_{j=1}^{m} \sum_{s\in K_j} \theta_s < 1$.
The throughput of node $i$, for $i=1,\ldots,n$, in a network with conflict graph $G(\vec \beta)$
matches $\theta_i$ if and only if the back-off rates are set according to \eqref{eq:nui_iline}.
\end{theorem}

While this results can be proven by induction similar to the proof in Appendix B, we skip this proof
as the result also follows from Theorem \ref{th:achieve_chordal} presented in the next section 
by noting that any generalized line network has a chordal conflict graph $G$ and 
$T=(\{K_1,$ $\ldots,K_m\},\mathcal{E})$ is a 
clique tree of $G$ when $\mathcal{E}=\{(j,j+1)|j=1,\ldots,m-1\}$.


\section{Chordal conflict graphs}\label{sec:chordal}

In this section we present the main result of the paper that unifies Theorem \ref{TH:TREE} and \ref{TH:LINE} 
as we present explicit formulas for
the back-off rates needed to achieve any achievable target throughput vector when the conflict graph $G$ is chordal.
A chordal graph $G=(V(G),E(G))$ is one in which all cycles consisting of more than $3$ nodes have a {\it chord}. A chord of a cycle 
is an edge joining two nonconsecutive nodes of the cycle. In other words, a chord is an edge that is not part of the cycle, but that 
connects two nodes belonging to the cycle. A graph is chordal if and only if it has a perfect elimination ordering \cite{blair1}, which is an
ordering of the nodes of the graph such that, for each $v \in V(G)$, $v$ and the neighbors of $v$ that occur after $v$ in the order form a clique. 
We will make use of this perfect elimination order further on, but present our main result using so-called clique trees first.

Let $\mathcal{K}_G = \{K_1,\ldots,K_m\}$ be the set of maximal cliques of $G$. 
A clique tree $T=(\mathcal{K}_G,\mathcal{E})$ 
is a tree in which the nodes correspond
to the maximal cliques and the edges are such that for any two maximal cliques $K$ and $K'$ the elements in $K \cap K'$
are part of any maximal clique on the unique path from $K$ to $K'$ in $T$.
The latter can be restated by demanding that the subgraph of $T$ induced by the maximal cliques that contain the node $v$
is a subtree of $T$ for any $v \in V$. 
 
An example of a chordal graph $G$ and possible clique tree $T$ is given in Figure \ref{fig:chordal_example}.
In this example several clique trees exist (e.g., replacing the edge between $K_4$ and $K_5$ by an edge
between $K_5$ and $K_6$ also produces a clique tree). 
One can show that a graph $G$ is chordal if and only if it has at least one clique tree (see Theorem 3.1 in \cite{blair1}).
A clique tree $T$ of a chordal graph can be constructed in linear time (by first computing a perfect elimination ordering, see
Section \ref{sec:algo}) and contains at most $V(G)$ nodes\footnote{In general a graph with $n$ nodes can have up to $3^{n/3}$
maximal cliques (and this upper bound is achieved by taking the complementary graph of the union of $n/3$ copies of 
the complete graph of size $3$) \cite{wood1}.}.

\begin{figure}[t]
\begin{center}
\begin{tikzpicture}[scale=0.9]
\tikzstyle{every node}=[circle, draw, fill=black!10,
                        inner sep=0pt, minimum width=10pt]

\node (1) at (1,0) []{$1$};
\node (2) at (2,0) []{$2$};
\node (3) at (3,0) []{$3$};
\node (4) at (4,0) []{$4$};
\node (5) at (5,0.5) []{$5$};
\node (6) at (4,1) []{$6$};
\node (7) at (3,1) []{$7$};
\node (8) at (2,1) []{$8$};
\node (9) at (1,1) []{$9$};
\node (10) at (2,2) []{$10$};
\node (11) at (3,2) []{$11$};

\draw [] (1) -- (2);
\draw [] (2) -- (3);
\draw [] (2) -- (7);
\draw [] (2) -- (8);
\draw [] (3) -- (4);
\draw [] (3) to[out=-45,in=-90] (5);
\draw [] (3) -- (6);
\draw [] (3) -- (7);
\draw [] (3) -- (8);
\draw [] (4) -- (5);
\draw [] (4) -- (6);
\draw [] (4) -- (7);
\draw [] (5) -- (6);
\draw [] (5) to[out=90,in=45] (7);
\draw [] (6) -- (7);
\draw [] (7) -- (8);
\draw [] (7) -- (10);
\draw [] (7) -- (11);
\draw [] (8) -- (9);
\draw [] (8) -- (10);
\draw [] (8) -- (11);

\node (k1) at (8,2.5) []{$K_1$};
\node (k2) at (7,1.5) []{$K_2$};
\node (k3) at (8,1.5) []{$K_3$};
\node (k4) at (8,0.5) []{$K_4$};
\node (k5) at (8,-0.5) []{$K_5$};
\node (k6) at (9,1.5) []{$K_6$};

\draw [] (k1) -- (k3);
\draw [] (k2) -- (k3);
\draw [] (k3) -- (k6);
\draw [] (k4) -- (k5);
\draw [] (k3) -- (k4);

\end{tikzpicture}
\end{center}
\caption{Example of a chordal graph $G$ with $n=11$ nodes (left) and one of its clique trees (right).
This graph contains $6$ maximal cliques $K_1=\{1,2\}, K_2=\{3,4,5,6,7\}, K_3=\{2,3,7,8\}, K_4=\{7,8,10\}, K_5=\{8,9\}$ and
$K_6=\{7,8,11\}$.}
\label{fig:chordal_example}
\end{figure}
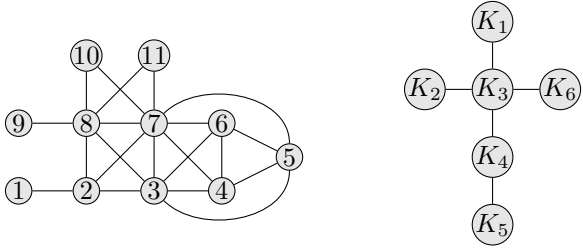

\subsection{Clique tree representation}

In this section we prove that if $\vec \theta$ is an achievable throughput vector for a network with a chordal
conflict graph $G$, we achieve $\vec \theta$ if and only if the back-off rates $\nu_i(\vec \theta)$ are given by 
\begin{align}\label{eq:nui_chordal2}
 \nu_i(\vec \theta) = \theta_i \
\frac{ \displaystyle \prod_{(K,K')\in \mathcal{E}, i \in K\cap K'} \left(1-\sum_{s \in K\cap K'} \theta_s\right)}
{ \displaystyle \prod_{K \in \mathcal{K}_G, i \in K}\left(1-\sum_{s \in K} \theta_s\right)},
\end{align}
for $i=1,\ldots,n$, where $T=(\mathcal{K}_G,\mathcal{E})$ is a clique tree of $G$.

For instance, for the example in Figure \ref{fig:chordal_example} we have
\begin{align*}
 \nu_2(\vec \theta) = \theta_2
\frac{1-\theta_2}{(1-\theta_1-\theta_2)(1-\theta_2-\theta_3-\theta_7-\theta_8)},
\end{align*}
as node $2$ belongs to $K_1=\{1,2\}$ and $K_3=\{2,3,7,8\}$ and $K_1 \cap K_3 = \{2\}$.\\

\noindent {\bf Remarks.} 

1) It may appear that the above rates depend on the clique tree $T$ that is used, as a chordal
graph does not necessarily have a unique clique tree. However, if $\mathcal{E}$ is the edge set of
any clique tree $T$ of $G$, then Theorem 4.2 of \cite{blair1} indicates that the {\it multiset} 
$\{K \cap K'| (K,K') \in \mathcal{E}\}$ is the same for every clique tree. As a result, we obtain the
same rate $\nu_i(\vec \theta)$ regardless of the clique tree used.\\

2) As for the acyclic and generalized line networks, the rate $\nu_i(\vec \theta)$ only depends on
the target throughputs $\theta_i \cup \{\theta_s | s \in \mathcal{N}_i\}$. In  other words the throughputs of
the non-interfering nodes have no impact on  $\nu_i(\vec \theta)$.\\

\begin{lemma}\label{lem:Zchordal}
Consider a network that has a chordal conflict graph $G$. 
Let $T=(\mathcal{K}_G,\mathcal{E})$ be a clique tree of $G$. Let $(K_l,K_m) \in \mathcal{E}$ and let
$(\mathcal{K}_l,\mathcal{E}_l)$ and $(\mathcal{K}_m,\mathcal{E}_m)$ be the two subtrees of $T$ obtained
after removing the edge $(K_l,K_m)$. Define $S = K_l \cap K_m$ and the sets $V_l$ and $V_m$ as
\[V_l = \left(\bigcup_{K \in \mathcal{K}_l} K \right) \setminus S,
\ \ \mbox{ and } \ \ V_m = \left(\bigcup_{K \in \mathcal{K}_m} K \right) \setminus S.\]
Let $Z_{V_l}(\vec \theta)$ be the normalizing constant for the network consisting of the
nodes in $V_l$ only  given that the back-off rate of node $i$ is set according to \eqref{eq:nui_chordal2} for $i=1,\ldots,n$. 
Then we have
\begin{align}\label{eq:chorZ}
Z_{V_l}(\vec \theta) &= \frac{\prod_{(K,K')\in \mathcal{E}_l}g_{K\cap K'}(\vec \theta) }
{\prod_{K \in \mathcal{K}_l} g_K(\vec \theta) } g_{K_l\cap K_m}(\vec \theta),
\end{align}
where $g_X(\vec \theta) = 1-\sum_{s \in X} \theta_s$.
\end{lemma}
\begin{proof}
By Lemma 4 in \cite{blair1}, we have that $S$, $V_l$ and $V_m$ form a partition of $V(G)$.
Let $T$ be a clique tree, $i \in V(G)$ and $T_i=(V(T_i),E(T_i))$ the graph induced by the cliques that contain $i$.
Clearly, $T_i$ is a subtree of $T$ (otherwise $T$ is not a clique tree). Let $i \in V_l$, meaning there
is some $K \in \mathcal{K}_l$ such that $i \in K$. Assume that $i \in K'$ with $K' \in \mathcal{K}_m$,
which implies that $(K_l,K_m)$ must be on the path from $K$ to $K'$. However, this means that $i \in S$,
which cannot be the case by definition of $V_l$. As a result, 
we may conclude that $V(T_i) \subset \mathcal{K}_l$.

We now prove the result by induction on $|\mathcal{K}_l|$. If $|\mathcal{K}_l|=1$, meaning $\mathcal{K}_l=\{K_l\}$, 
the nodes in $V_l$ only belong to the maximal clique $K_l$. Thus,
\begin{align*}
 Z_{V_l}(\vec \theta) &= 1 + \sum_{i \in V_l} \nu_i(\vec \theta) = 1+ \sum_{i \in V_l} \frac{\theta_i}{1-\sum_{s \in K_l} \theta_s} \\
& =\frac{1- \sum_{s \in K_l \cap K_m} \theta_s}{1-\sum_{s \in K_l} \theta_s} = \frac{g_{K_l\cap K_m}(\vec \theta)}{g_{K_l}(\vec \theta)}, 
\end{align*} 
as $K_l$ is the union of the disjoint sets $V_l$ and $S=K_l \cap K_m$. 
For $|\mathcal{K}_l| > 1$, we introduce the following notations. Let $\mathcal{N}_K$ be the neighbors of
$K \in \mathcal{K}_G$ in $T$ and $\mathcal{W}_i$ be the subset of $\mathcal{E}$ holding the edges that
connect a maximal clique that contains $i$ with one in $\mathcal{K}_l$ that does not, that is,  
\[\mathcal{W}_i = \{(K,K') | (K,K') \in \mathcal{E},
K \in \mathcal{K}_l \setminus V(T_i), K' \in  V(T_i)\}.\] 
We now use the fact that at most one node can be active simultaneously in the set $V_l \cap K_l$ to argue that
\begin{align*} 
Z_{V_l}(\vec \theta) &= \sum_{i \in V_l \cap K_l} \nu_i(\vec \theta)\prod_{(K_{w},K_z) \in \mathcal{W}_i} Z_{V_{w}}(\vec \theta)\\
&+ \prod_{K_{j} \in \mathcal{N}_{K_l} \setminus \{K_m\}} Z_{V_{j}}(\vec \theta),
\end{align*} 
where $V_j$ and $V_w$ are defined as $V_l$ if we replace $(K_l,K_m)$ by $(K_j,K_l)$ and $(K_w,K_z)$, respectively. 
Note that $V_j \subset V_l$ as any $v\in K_l \cap K_m$ that is part of $\bigcup_{K \in \mathcal{K}_j} K$ must
be part of $K_j \cap K_l$ and similarly $V_w \subset V_l$. Further, the sets $V_j$ are pair-wise disjoint (as
are the sets $V_w$) allowing us to take the products of the normalizing constants. 

By induction and due to $V(T_i) \subset \mathcal{K}_l$, for $i \in V_l$, we have
\begin{align*}
\MoveEqLeft Z_{V_l}(\vec \theta) = \sum_{i \in V_l \cap K_l} \theta_i 
\frac{\prod_{(K,K')\in \mathcal{E}_l, i \in K\cap K'} g_{K\cap K'}(\vec \theta)}
{ \prod_{K \in \mathcal{K}_l, i \in K} g_K(\vec \theta)} \\
&\hspace*{-0.5cm} \times \frac{\prod_{(K,K')\in \mathcal{E}_l, i \not\in K\cap K'}g_{K\cap K'}(\vec \theta) }
{\prod_{K \in \mathcal{K}_l, i \not\in K} g_{K}(\vec \theta) }
+ \frac{\prod_{(K,K')\in \mathcal{E}_l} g_{K\cap K'}(\vec \theta)}
{ \prod_{K \in \mathcal{K}_l \setminus \{K_l\}} g_K(\vec \theta)}\\
&\hspace*{-0.5cm} = \frac{\prod_{(K,K')\in \mathcal{E}_l} g_{K\cap K'}(\vec \theta)}
{ \prod_{K \in \mathcal{K}_l } g_K(\vec \theta)} 
\underbrace{\left( \sum_{i \in V_l \cap K_l} \theta_i + g_{K_l}(\vec \theta)\right)}_{1-\sum_{s \in K_l \cap K_m} \theta_s},
\end{align*}
which completes the proof.
\end{proof}

For instance, for the example in Figure \ref{fig:chordal_example} if we let $(K_l,K_m) = (K_4,K_3)$, we find that
$\mathcal{K}_l = \{K_4,K_5\}$, $\mathcal{E}_l = \{(K_4,K_5)\}$, $V_l = \{9,10\}$ and the above lemma states that
\begin{align*}
Z_{V_l}(\vec \theta)  &= \frac{g_{K_4\cap K_5}(\vec \theta)}{g_{K_5}(\vec \theta)g_{K_4}(\vec \theta)} g_{K_4 \cap K_3}(\vec \theta)\\
&= \frac{(1-\theta_8)(1-\theta_7-\theta_8)}{(1-\theta_8-\theta_9)(1-\theta_7-\theta_8-\theta_{10})}\\
&= \left(1+\frac{\theta_9}{1-\theta_8-\theta_9}\right)\left(1+\frac{\theta_{10}}{1-\theta_7-\theta_8-\theta_{10}}\right) \\
&= (1+\nu_9(\vec \theta))(1+\nu_{10}(\vec \theta)),
\end{align*}
as expected.

\begin{lemma}\label{lem:Zchordal2}
Consider a network that has a chordal conflict graph $G$. 
Let $Z_n(\vec \theta)$ be the normalizing constant for the network
 given that the back-off rate of node $i$ is set according to \eqref{eq:nui_chordal2} for $i=1,\ldots,n$, then
\begin{align}\label{eq:chorZn}
Z_{n}(\vec \theta) &= \frac{\prod_{(K,K')\in \mathcal{E}}\left(1-\sum_{s \in  K \cap K'} \theta_s \right) }
{\prod_{K \in \mathcal{K}_G} \left(1-\sum_{s \in  K} \theta_s \right) }.
\end{align}
\end{lemma}
\begin{proof}
Let $K_m$ be a leaf in the clique tree $T=(\mathcal{K}_G,\mathcal{E})$ and $(K_l,K_m) \in \mathcal{E}$. Clearly
$Z_{n}(\vec \theta)$ can be written by separating the terms where one of the nodes in $K_m$ is
active and those where these nodes are silent. By using similar reasoning than in Lemma \ref{lem:Zchordal}, we get 
\[Z_{n}(\vec \theta) = \sum_{i \in K_m} \nu_i(\vec \theta) \prod_{(K_{w},K_z) \in \mathcal{U}_i} Z_{V_{w}}(\vec \theta) + Z_{V_l}(\vec \theta), \]
as $V_l = \{1,\ldots,n\} \setminus K_m$, where 
\[\mathcal{U}_i = \{(K,K') | (K,K') \in \mathcal{E},
K \in \mathcal{K}_G \setminus V(T_i), K' \in  V(T_i)\}.\] 
This yields
\begin{align*}
 Z_{n}(\vec \theta) &= \sum_{i \in K_m} \theta_i 
\frac{\prod_{(K,K')\in \mathcal{E}, i \in K\cap K'} g_{K\cap K'}(\vec \theta)}
{ \prod_{K \in \mathcal{K}_G, i \in K} g_K(\vec \theta)} \\
&\hspace*{1cm} \times \frac{\prod_{(K,K')\in \mathcal{E}, i \not\in K\cap K'}g_{K\cap K'}(\vec \theta) }
{\prod_{K \in \mathcal{K}_G, i \not\in K} g_{K}(\vec \theta) }\\
&+ \frac{\prod_{(K,K')\in \mathcal{E}} g_{K\cap K'}(\vec \theta)}
{ \prod_{K \in \mathcal{K}_G \setminus \{K_m\}} g_K(\vec \theta)} \\
&= \frac{\prod_{(K,K')\in \mathcal{E}} g_{K\cap K'}(\vec \theta)}
{ \prod_{K \in \mathcal{K}} g_K(\vec \theta)} 
\underbrace{\left( \sum_{i \in K_m} \theta_i + g_{K_m}(\vec \theta)\right)}_{= 1},
\end{align*}
by Lemma \ref{lem:Zchordal} as $\mathcal{K}_l = \mathcal{K}_G \setminus \{K_m\}$.
\end{proof}

\begin{theorem}\label{th:achieve_chordal} 
Consider a network with a chordal conflict graph $G$.
Let $\vec \theta = (\theta_1,\ldots,\theta_n)$ be a positive vector with $T=\max_{j \in \mathcal{K}_G} \sum_{s\in K_j} \theta_s < 1$.
The throughput of node $i$, for $i=1,\ldots,n$, in a network with conflict graph $G$
matches $\theta_i$ if and only if the back-off rates are set according to \eqref{eq:nui_chordal2}.
\end{theorem}
\begin{proof}
The proof follows from Lemma \ref{lem:Zchordal} and \ref{lem:Zchordal2} by noting that the throughput of
node $i$ can be written as
\[ \nu_i(\vec \theta)   \prod_{(K_{w},K_z) \in \mathcal{U}_i} Z_{V_{w}}(\vec \theta)/Z_n(\vec \theta),\]
where 
\[\mathcal{U}_i = \{(K,K') | (K,K') \in \mathcal{E},
K \in \mathcal{K}_G \setminus V(T_i), K' \in  V(T_i)\},\] 
as before.
\end{proof}

\noindent {\bf Remarks.}

1) Theorem \ref{th:achieve_chordal} implies that $\Gamma$, the set of achievable
throughput vectors defined in \eqref{eq:Gamma}, can be expressed as
\[\Gamma =\{(\theta_1,\ldots,\theta_n)| \max_{j \in \mathcal{K}_G} \sum_{s\in K_j} \theta_s < 1\}.\]
Hence, \eqref{eq:nui_chordal2} can be applied for any achievable throughput vector when the network 
has a chordal conflict graph.\\

2) It is easy to see that \eqref{eq:nui_chordal2} coincides with \eqref{eq:nui_iline} 
in case the network is a generalized line network as $T=(\mathcal{K}_G,\mathcal{E})$ with
$\mathcal{E} =\{(j,j+1)|j=1,\ldots,m-1\}$, is a clique tree (in fact, it is the unique clique tree). 
If the conflict graph is acyclic \eqref{eq:nui_chordal2} reduces to \eqref{eq:nui_tree} as
the maximal cliques correspond to the edges $E(G)$ and the subgraph of $T$ induced by
the maximal cliques that contain $i \in V(G)$ forms a subtree $T_i$ of $T$ (as otherwise $T$ is not a clique tree).
The subtree $T_i$ contains exactly $|\mathcal{N}_i|-1$ edges $(K,K') \in \mathcal{E}$ and $K \cap K' = \{i\}$ for these
edges.\\

3) Let $k_i$ be the number of nodes in $T_i$ (the subtree of $T$ induced by the cliques containing $i$) and
 $k_{i,j}$ with $j \in \mathcal{N}_i$ the number of nodes of the subtree induced by the cliques containing
 both $i$ and $j$.  As $\Vert \vec \theta \Vert \downarrow 0$ we have
\begin{align*}
\nu_i(\vec \theta) &= \theta_i \left[1-\sum_{(K,K')\in \mathcal{E}, i \in K\cap K'} \sum_{s\in K\cap K'} \theta_s +O(\Vert \vec \theta \Vert^2)\right] \\
& \hspace*{2cm} \times  \left[1+\sum_{K \in \mathcal{K}_G, i \in K}\sum_{s \in K} \theta_s + O(\Vert \vec \theta \Vert^2)\right]\\
&= \theta_i \left[1-(k_i-1)\theta_i-\sum_{j \in \mathcal{N}_i}(k_{i,j}-1)\theta_j+O(\Vert \vec \theta \Vert^2)\right] \\
& \hspace*{2cm} \times  \left[1+k_i\theta_i-\sum_{j \in \mathcal{N}_i}k_{i,j}\theta_j+ O(\Vert \vec \theta \Vert^2)\right]\\
&= \theta_i \left(1+ \theta_i + \sum_{j \in \mathcal{N}_i} \theta_j\right)+O(\Vert \vec \theta \Vert^3),
\end{align*}
which is in agreement with the light-traffic approximation for general conflict graphs in Proposition 2 of \cite{vandeven5}.

\subsection{Perfect elimination ordering representation}\label{sec:algo}
In this subsection we show that the back-off rates as specified by \eqref{eq:nui_chordal2} can also be obtained using 
Algorithm \ref{algo:nui} if the network has a chordal conflict graph $G$. The first step of this algorithm exists
in determining a perfect elimination ordering of $G$, which can be achieved in $O(|V(G)|+|E(G)|)$ time by relying either on the
lexicographic breadth-first search algorithm \cite{rose1} or on the maximum cardinality search (MCS) algorithm introduced in
\cite{tarjan1}. The MCS algorithm determines the perfect elimination ordering by picking $\alpha(n)$ at random
and subsequently determines $\alpha(i)$ by selecting the node with the most neighbors in $\{\alpha(i+1),\ldots,\alpha(n)\}$,
breaking ties arbitrarily. For instance, for the example in Figure \ref{fig:chordal_example} the perfect elimination
order $(9,11,10,6,5,4,8,7,3,2,1)$ can be obtained in this manner.

After determining this ordering Algorithm \ref{algo:nui} operates such that after executing the main for loop 
from $n-1$ down to $j$, the rates $\nu_i(\vec \theta)$ are such that the throughput of node $\alpha(i)$, for $i \geq j$, is 
$\theta_i$ if we consider the network consisting of the nodes $\alpha(j),\ldots,\alpha(n)$ only.  In other words,
if we extend the network consisting of the nodes $\alpha(i+1),\ldots,\alpha(n)$ by node $\alpha(i)$, it suffices to adapt
the rate of all the neighbors of $\alpha(i)$ in the set $\{\alpha(i+1),\ldots,\alpha(n)\}$, which is denoted as
$\mathcal{M}_{\alpha(i)}$,  by the same factor and to
set the rate of $\alpha(i)$ as if it is part of a complete conflict graph consisting of the nodes $\{\alpha(i)\} \cup \mathcal{M}_{\alpha(i)}$
only.

\begin{algorithm}[t]
  \KwIn{A chordal conflict graph $G$} %
  \KwOut{Back-off rates $\nu_1(\vec \theta),\ldots,\nu_n(\vec \theta)$} %
  Determine a perfect elimination ordering  of $G$\\
	\For{$i=1$ to $n$ }{
	Let $\alpha(i)$ be the node in position $i$ in this order\;
	}
	\For{$i=1$ to $n$ }{
	Let $\mathcal{M}_{\alpha(i)} = \mathcal{N}_{\alpha(i)} \cap \{\alpha(i+1),\ldots,\alpha(n)\}$\;
	}
	
	$\nu_{\alpha(n)}(\vec \theta)=\theta_{\alpha(n)}/(1-\theta_{\alpha(n)})$\;
	\For{$i=n-1$ {\bf down to} $1$ }{
    $\nu_{\alpha(i)}(\vec \theta)=\theta_{\alpha(i)}/(1-\theta_{\alpha(i)}-\sum_{s \in \mathcal{M}_{\alpha(i)}} \theta_s)$\;
		\For{$j\in \mathcal{M}_{\alpha(i)}$}{
		$\nu_j(\vec \theta)=\nu_j(\vec \theta)\frac{1-\sum_{s \in \mathcal{M}_{\alpha(i)}} \theta_s}
		{1-\theta_{\alpha(i)}-\sum_{s \in \mathcal{M}_{\alpha(i)}} \theta_s}$\;
		}
	}
  \caption{\bf Algorithm to determine the back-off rates for a given achievable throughput vector $\vec \theta$ in a network with
	a chordal conflict graph $G$.
  }\label{algo:nui}
\end{algorithm}

The outcome of Algorithm \ref{algo:nui} can be expressed as follows:
\begin{align}\label{eq:nui_chordal}
\nu_{\alpha(i)}(\vec \theta) = \theta_{\alpha(i)} \
\frac{ \displaystyle \prod_{j<i, \alpha(i) \in \mathcal{N}_{\alpha(j)}}\left(1-\sum_{s \in \mathcal{M}_{\alpha(j)}} \theta_s\right)}
{ \displaystyle \prod_{j\leq i, \alpha(i) \in \mathcal{N}_{\alpha(j)}^+}\left(1-\theta_{\alpha(j)}-\sum_{s \in \mathcal{M}_{\alpha(j)}} \theta_s\right)},
\end{align}
where $\mathcal{N}_{i}^+ = \mathcal{N}_{i} \cup \{i\}$ and $\nu_{\alpha(i)}(\vec \theta) $ is the required back-off rate for the node 
in position $i$ in the perfect elimination ordering.

\begin{prop}\label{prop:peo}
The back-off rates specified by \eqref{eq:nui_chordal} are equal to \eqref{eq:nui_chordal2} irrespective of the perfect elimination
ordering used.
\end{prop}
\begin{proof}
We start by defining a set of $n$ trees recursively, such that tree $T_\alpha^{(j)}$ is a clique tree for the
graph $G_\alpha^{(j)}$ induced by $\{\alpha(j),\ldots,\alpha(n)\}$. Let $T_\alpha^{(n)}$ be a tree consisting of a single
node labeled $\{\alpha(n)\}$, clearly this is a clique tree for the graph $G_\alpha^{(n)}$. 
Given that $T_\alpha^{(j+1)}$ is a clique tree for the graph $G_\alpha^{(j+1)}$,
we now define $T_\alpha^{(j)}$, for $j=n-1$ down to $1$.  
The nodes $\mathcal{M}_{\alpha(j)}$ form a clique (as $\alpha$ is a perfect elimination ordering), therefore
there exists a maximal clique $K$ in $G_\alpha^{(j+1)}$ such that $\mathcal{M}_{\alpha(j)} \subset K$. Either
$K = \mathcal{M}_{\alpha(j)}$, in which case we obtain a clique tree $T_\alpha^{(j)}$ for $G_\alpha^{(j)}$ by replacing the node $K$
of $T_\alpha^{(j+1)}$ by $K'= K \cup \{\alpha(j)\}$ (as $K$ is not a maximal clique in $G_\alpha^{(j)}$). Otherwise, we construct $T_\alpha^{(j)}$
by adding a node $K'=\{\alpha(j)\} \cup \mathcal{M}_{\alpha(j)}$ to $T_\alpha^{(j+1)}$ and connect $K'$ with $K$.
   
We now show that the $\nu_i(\vec \theta)$ values obtained by Algorithm \ref{algo:nui} after executing the main for
loop down to $j$ are equal to 	 
\begin{align*}
x_i^{(j)}  = \theta_i \
\frac{\prod_{(K,K')\in E(T_\alpha^{(j)}), i \in K\cap K'} \left(1-\sum_{s \in K\cap K'} \theta_s\right)}
{ \prod_{K \in V(T_\alpha^{(j)}), i \in K}\left(1-\sum_{s \in K} \theta_s\right)},
\end{align*}
for $i \in \{\alpha(j),\ldots,\alpha(n)\}$,
which suffices to prove the proposition as $T^{(1)}_\alpha$ is a clique tree for $G$.  
For $j=n$ this is true as $\nu_{\alpha(n)}(\vec \theta) = \theta_{\alpha(n)}/(1-\theta_{\alpha(n)}) = x_{\alpha(n)}^{(n)}$.

Assume $\nu_i(\vec \theta)$, for $i \in \{\alpha(j+1),\ldots,\alpha(n)\}$, is equal to $x_i^{(j+1)}$ after executing
the main for loop down to $j+1$. When $i=j$, the rate $\nu_i(\vec \theta)$, with $i \in \mathcal{M}_{\alpha(j)}$, are multiplied by
$g_{\mathcal{M}_{\alpha(j)}}(\vec \theta)/g_{\mathcal{M}_{\alpha(j)}\cup \{\alpha(j)\}}(\vec \theta)$,
while $\nu_{\alpha(j)}(\vec \theta)$ is set equal to $\theta_{\alpha(j)}/g_{\mathcal{M}_{\alpha(j)}\cup \{\alpha(j)\}}(\vec \theta)$ 
and the other rates remain the same. By construction of $T_\alpha^{(j)}$ the only maximal clique $K$ in $G_\alpha^{(j)}$ that contains
$\alpha(j)$ is $K=\{\alpha(j)\} \cup \mathcal{M}_{\alpha(j)}$, meaning $x_{\alpha(j)}^{(j)} = \nu_{\alpha(j)}(\vec \theta)$.

For $i \in \mathcal{M}_{\alpha(j)}$ we have two cases. Either $V(T_\alpha^{(j)})=V(T_\alpha^{(j+1)})$ in which case
the set of edges $\{(K,K') | i \in K \cap K'\}$ in $T_\alpha^{(j)}$ and $T_\alpha^{(j+1)}$ are the same, while
the clique $K=\mathcal{M}_{\alpha(j)}$ is replaced by the clique $K'=\mathcal{M}_{\alpha(j)} \cup \{\alpha(j)\}$.
Hence, 
\begin{align}\label{eq:xij}
x_i^{(j)}=x_i^{(j+1)}g_{\mathcal{M}_{\alpha(j)}}(\vec \theta)/g_{\mathcal{M}_{\alpha(j)}\cup \{\alpha(j)\}}(\vec \theta).
\end{align}
When $V(T_\alpha^{(j)})\not=V(T_\alpha^{(j+1)})$, the tree $T_\alpha^{(j)}$ is identical to $T_\alpha^{(j+1)}$,
except that it contains an extra node $K'=\mathcal{M}_{\alpha(j)} \cup \{\alpha(j)\}$ 
and edge $(\mathcal{M}_{\alpha(j)},K')$, which yields \eqref{eq:xij}. 

Finally, by construction of $T_\alpha^{(j)}$, it is easy to see that $x_i^{(j)}=x_i^{(j+1)}$ for $i \not\in
\mathcal{M}_{\alpha(j)}\cup \{\alpha(j)\}$. 
\end{proof}

\subsection{Distributed algorithm}\label{sec:dist}
Algorithm \ref{algo:nui} can be used to determine the back-off rates in a centralized manner. In this section we indicate that
a node can determine its back-off rate obtained by Algorithm \ref{algo:nui} 
in a fully distributed manner with limited message passing. More specifically,
it suffices for node $i$ to discover its set of neighbors $\mathcal{N}_i$ in the conflict graph $G$, their target throughputs 
$\{\theta_j | j \in \mathcal{N}_i\}$ 
as well as the set of neighbors $\mathcal{N}_j$ for each $j \in \mathcal{N}_i$. With this information node $i$ can construct the subgraph $G[\mathcal{N}_i^+]$ 
of the conflict graph $G$ induced by the nodes in $\mathcal{N}_i^+$. 

To obtain the back-off rate $\nu_{i}(\vec \theta)$ node $i$
executes Algorithm  \ref{algo:nui} on the graph $G[\mathcal{N}_i^+]$ and sets its own back-off rate accordingly. Note that 
when node $i$ runs Algorithm  \ref{algo:nui} locally on the conflict graph $G[\mathcal{N}_i^+]$, it computes a back-off
rate for all the nodes in  $\mathcal{N}_i^+$. However, only the rate computed for node $i$ itself is of interest and used. 
The next proposition proves the correctness of this distributed algorithm:

\begin{prop}
The back-off rate for node $i \in V$ given by executing Algorithm \ref{algo:nui} on the conflict graph
$G$ is identical to the rate obtained for node $i$ when executing Algorithm \ref{algo:nui} on the subgraph $G[\mathcal{N}_i^+]$
of $G$ induced by the nodes in  $\mathcal{N}_i^+$.
\end{prop}
\begin{proof}
Let $\alpha$ be the perfect elimination order (peo) used when executing Algorithm \ref{algo:nui} on $G$. 
Clearly, if we order the nodes in  $\mathcal{N}_i^+$ in the same manner we obtain a peo for $G[\mathcal{N}_i^+]$
and therefore $G[\mathcal{N}_i^+]$ is chordal. We now argue that if node $i$  executes Algorithm \ref{algo:nui} locally
on $G$ using this same peo for the nodes in $\mathcal{N}_i^+$, node $i$ obtains the same rate as when
performing Algorithm \ref{algo:nui} on $G$. As the peo used for $G[\mathcal{N}_i^+]$ has no impact on the rates obtained 
by Algorithm \ref{algo:nui} (due to Proposition \ref{prop:peo}), this suffices to complete the proof. 

The rate of node $i$ is initially set equal to $\theta_i / (1-\theta_i-\sum_{s \in \mathcal{M}_i} \theta_s)$ when 
Algorithm \ref{algo:nui} is executed on $G$ or $G[\mathcal{N}_i^+]$.  This rate is updated whenever we encounter a node
$j \in \mathcal{N}_i$ that appears before node $i$ in the peo. When encountering such a node, the rate of node $i$ is multiplied by
  $(1-\sum_{s \in \mathcal{M}_j} \theta_s)/(1-\theta_j-\sum_{s \in \mathcal{M}_j} \theta_s)$ when running 
	Algorithm \ref{algo:nui} on $G$ and by $(1-\sum_{s \in \mathcal{M}_j \cap \mathcal{N}_i^+} \theta_s)/
	(1-\theta_j-\sum_{s \in \mathcal{M}_j\cap \mathcal{N}_i^+ } \theta_s)$ when running it on $G[\mathcal{N}_i^+]$. 
	As $j$ appears before $i$ and $j$ is a neighbor of $i$, all the nodes appearing
	after $j$ that belong to $\mathcal{N}_j$ are also neighbors of $i$ (otherwise the order is not a peo). 
	In other words all the nodes in $\mathcal{M}_j$ are also part of $G[\mathcal{N}_i^+]$ and $\mathcal{M}_j\cap \mathcal{N}_i^+  = \mathcal{M}_j$.
\end{proof}

\section{Chordal approximations}\label{sec:approx}
The purpose of this section is twofold. First, using $2$ arbitrarily chosen chordal conflict graphs 
with $100$ nodes, we confirm the correctness of our main result by simulation.
Second, we investigate the
possibility of using our main result to construct approximate back-off rates for general conflict graphs.

The first of the $2$ randomly chosen examples  (see Figure \ref{fig:n100})
contains $70$ maximal cliques, the largest one containing only $7$ nodes. The nodes have
on average $2.58$ conflicting nodes, while node $99$ has $51$ conflicts. We simulated the
corresponding Markov chain
up to time $t=10^7$ using \eqref{eq:nui_chordal} such that the target throughput of each link matches $1/20$.
The lowest and highest observed throughputs, computed as the number of packets transmitted divided by $t$, 
among the $100$ nodes were $0.04983$ and $0.05014$, respectively. Note that we selected $1/20$
as higher values require higher back-off rates which slows down the convergence of the Markov chain to its
steady state. For instance, if we aim at a fair throughput of $1/8$, the back-off rate of node $99$
should be set equal to   $3038765625/131072 \approx 23184$ according to \eqref{eq:nui_chordal}, meaning
node $99$ will have periods were it can transmit many packets in a row, followed by long periods of inactivity.

\begin{figure}[t]
\center
\includegraphics[width=0.22\textwidth]{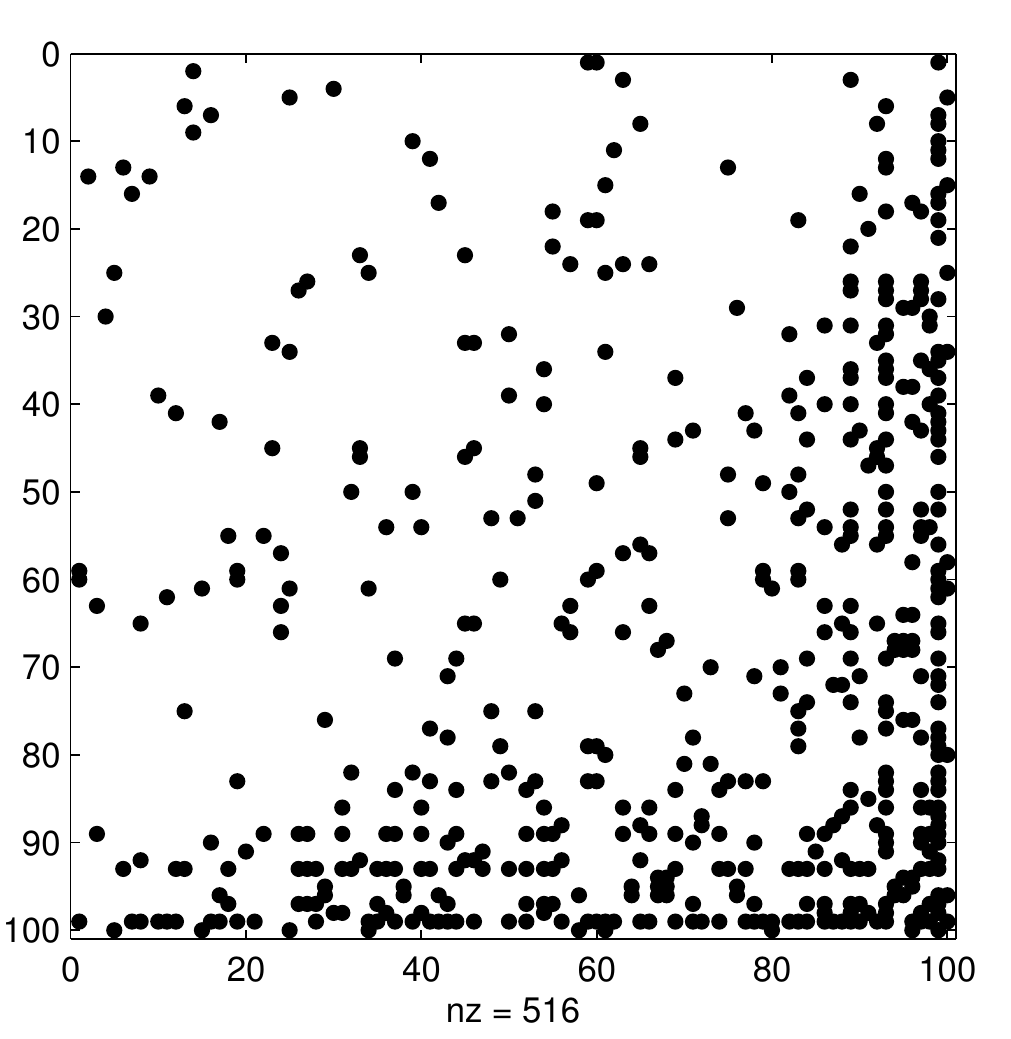}
\includegraphics[width=0.22\textwidth]{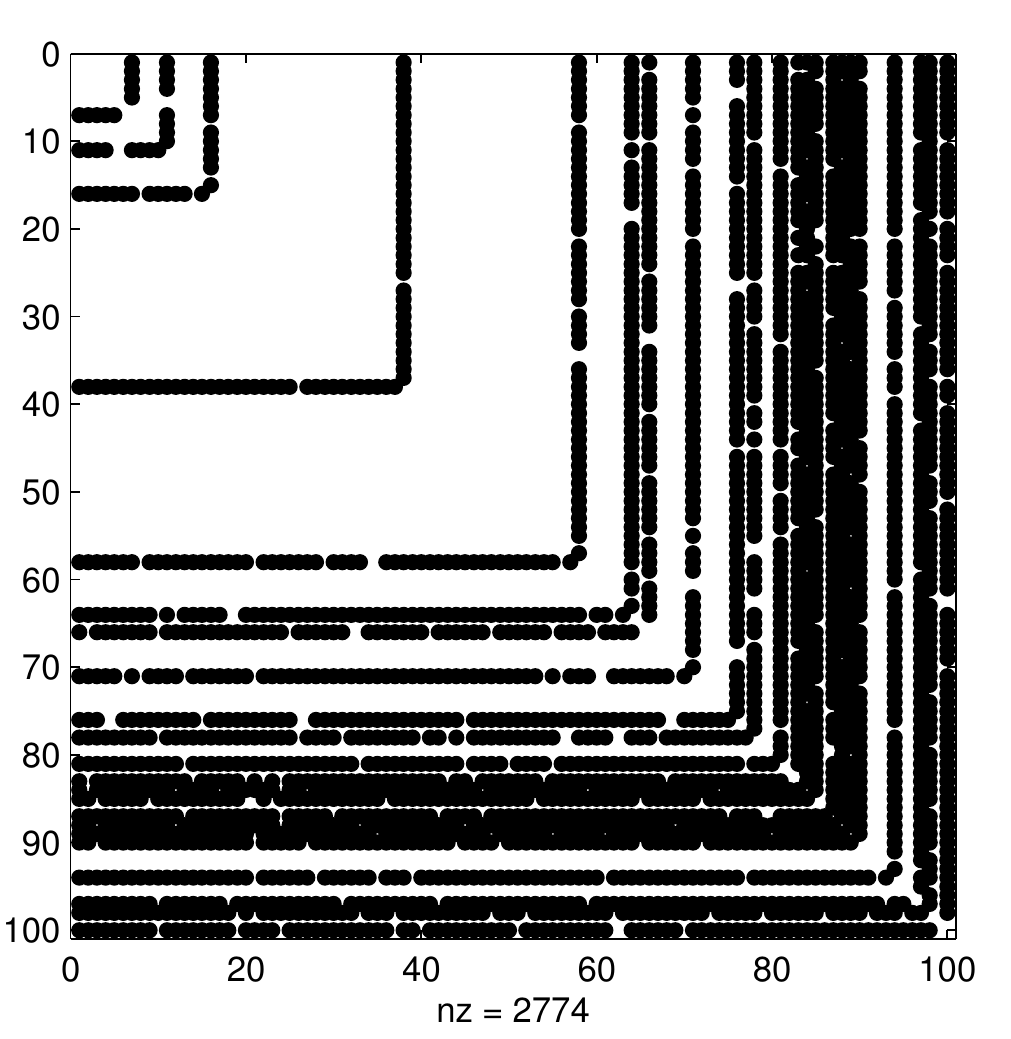}
\caption{Adjacency matrix $A$ of the $2$ chordal conflict graphs with $n=100$ nodes. $A(i,j) = 1$ if
$(i,j) \in E$ and $0$ otherwise. The number of non-zero entries in $A$ is denoted as nz.}
\label{fig:n100}
\end{figure}

The second size $100$ conflict graph considered is of a different nature (see Figure \ref{fig:n100}) as it is more dense 
than the first, it contains $79$ maximal cliques, where the largest one has size $23$. Simulating the Markov chain (up to time $t=10^7$)
with the target throughput of each node set to $1/50$ resulted in a lowest and highest observed throughput of $0.01990$ and $0.02012$,
respectively. 



An interesting question at this point is whether we can define a {\it chordal} approximation 
for the rates $\nu_i(\vec \theta)$ when the conflict graph $G$ is {\it not} chordal.
To study the accuracy of the proposed approximations we considered a class of conflict graphs obtained by placing a set of $n$ nodes 
in a random manner in a square of size $1$ and assumed that node $i$ and $j$ are in conflict when the Euclidean distance 
between node $i$ and $j$ was below some threshold $R$, e.g., $0.2$ (see Figure \ref{fig:network}). 
For the experiments conducted we considered cases with $n=50$ and $n=100$ nodes. 

\begin{figure}[t]
\center
\includegraphics[width=0.3\textwidth]{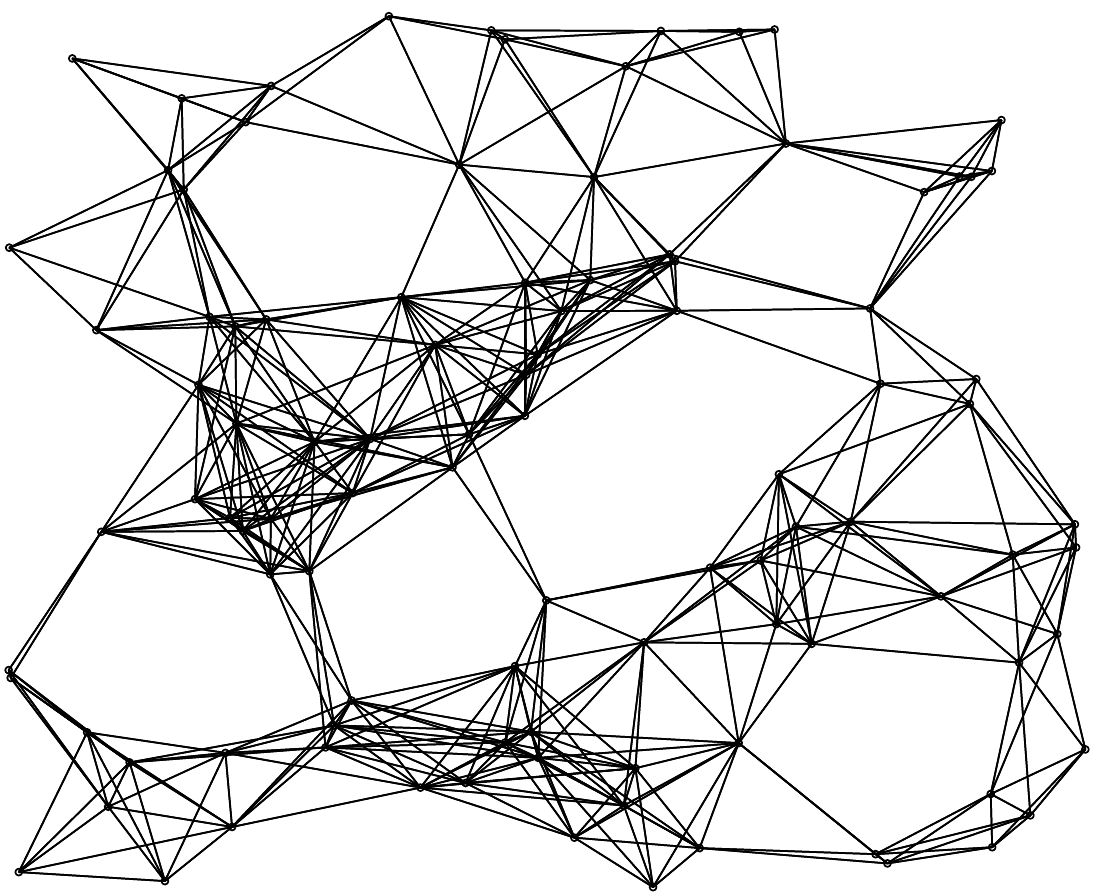}
\caption{Non-chordal conflict graph with $n=100$ nodes and $485$ edges.}
\label{fig:network}
\end{figure}

A first option is to note that any graph has a {\it chordal completion}, that is, we can make any graph $G$ chordal by
adding a set $E'$ of edges to $G$. While the problem of finding the set $E'$ of minimal cardinality
(known as the minimum fill) or the set $E'$ that minimizes the size of the maximum clique (known as
minimum treewidth) are NP-complete \cite{yannakakis1}, the former can be approximated using a simple {\it minimum degree} 
heuristic \cite{heggernes1}. For instance, when applied to the conflict graph in Figure \ref{fig:network}, the set $E'$
contains $273$ edges, which indicates that the conflict graph is far from being chordal.
The general observation of these experiments was that the chordal completion
approximation is {\it far less} accurate than the earlier mentioned Bethe approximation.
The main reason is that the degree of some nodes is substantially larger in the chordal completion and therefore 
the approximation suggests to use back-off rates that are much larger than what is needed in the original network.




We should note that the minimum degree heuristic of \cite{heggernes1} does not necessarily produce a set $E'$ that is minimal and
efficient algorithms that can reduce $E'$ to a minimal set have been developed. Hence, we could apply 
one of these algorithms to reduce $E'$ to a minimal set $E''$, such that $(V(G),E(G)\cup E'')$ is chordal.
This should improve the accuracy of the chordal completion approximation. 
However, experiments in \cite{blair2} have shown that the minimum degree heuristic often produces sets $E'$ that
are nearly minimal and as such the improvement is likely to be very minor.

\begin{algorithm}[t]
  \KwIn{A general conflict graph $G=(V,E)$} %
  \KwOut{A maximal chordal subgraph $\tilde G=(V,\tilde E)$ and peo $\alpha$} %
  \For{$v \in V$ }{
	$C(v) = \emptyset$\;
	}
	$k=|V|$; $\tilde E=\emptyset$\;
	Select any $v_0 \in V$; set $S_k=\{v_0\}$; $\alpha(k) = v_0$\;
	\For{$u \in V \setminus S_k$ with $(u,v_0) \in E$}{
		\If{$C(u) \subseteq C(v_0)$}{
		$C(u)=C(u)\cup\{v_0\}$;$\tilde E=\tilde E\cup (u,v_0)$;}
	}
Let $v_0 \in V \setminus S_k$ with $|C(v_0)|\geq |C(v)|$ for $v \in V \setminus S_k$\;
	Set $\alpha(k-1)=v_0$; $S_{k-1}=S_k \cup \{v_0\}$; $k=k-1$\;
	\If{$k>1$}{Go to line 6;} 
  \caption{\bf MAXCHORD algorithm of \cite{dearing1}.
  }\label{algo:MAXCHORD}
\end{algorithm}

Another manner to construct a chordal approximation exists in constructing a chordal subgraph of $G$ and using the back-off rates
of the subgraph in the original network. For this purpose we can rely on the MAXCHORD algorithm introduced in \cite{dearing1} (see Algorithm \ref{algo:MAXCHORD}).
This algorithm is a variation on the MCS algorithm discussed earlier, which
given a graph $G$, computes a maximal chordal subgraph $\tilde G=(V(G),\tilde E)$ 
with $\tilde E \subset E(G)$.  As the node $v_0$ in line $5$ we selected a node of maximum degree and the same was done on line $11$
of the algorithm when breaking ties.

Numerical experiments indicated that the chordal subgraph approximation is a lot more accurate than the chordal
completion approximation, but is still less accurate compared to the Bethe approximation. 
The inaccuracy is caused by the fact that some nodes use back-off
rates that are too small, leading to a throughput below target, as they have more conflicting nodes in the original network
than in the chordal subgraph. Further, both the chordal completion and subgraph approximations do not have an obvious distributed implementation.

\begin{table}
\center
\begin{tabular}{ l | c c| c c |c c }
 \hline
& \multicolumn{2}{c|}{$R=0.15,$} & \multicolumn{2}{c|}{$R=0.2,$} & \multicolumn{2}{c}{$R=0.25,$}\\
& \multicolumn{2}{c|}{$k_{max}=7$} & \multicolumn{2}{c|}{$k_{max}=9$} & \multicolumn{2}{c}{$k_{max}=12$}\\ 
  $\theta_i$ & LCS & Bethe & LCS & Bethe & LCS & Bethe \\ \hline
 $\sfrac{0.45}{k_{max}}$ & $0.23\%$ & $3.48\%$ & $0.40\%$&$5.68\%$ &$0.74\%$&$6.98\%$\\
$\sfrac{0.55}{k_{max}}$ & $0.35\%$ & $4.88\%$ & $0.65\%$& $7.78\%$&$1.20\%$&$9.32\%$ \\
 $\sfrac{0.65}{k_{max}}$ & $0.54\%$ & $6.37\%$ & $1.00\%$&$9.94\%$ &$2.02\%$&$11.72\%$\\
 $\sfrac{0.75}{k_{max}}$ & $0.93\%$ & $7.92\%$ & $1.61\%$& $12.16\%$&$3.56\%$&$14.05\%$\\
 $\sfrac{0.85}{k_{max}}$& $1.52\%$ & $9.52\%$ & $2.66\%$& $14.37\%$& $6.64\%$&$16.34\%$\\ 
\hline
\end{tabular}
\vspace*{1mm}
\caption{Average relative deviation from the target throughput in a simulation run of length $10^7$ for the
local chordal subgraph (LCS) and Bethe approximation, where $k_{max}$ is the maximum clique size of the conflict graph.}
\label{tab:approx_res_locsub}  
\end{table}


The solution in developing a more accurate and distributed chordal approximation exists in letting node $i$ determine
its back-off rate by computing a maximal chordal subgraph $G_i$ of $G[\mathcal{N}_i^+]$, the subgraph of $G$ induced by 
$\mathcal{N}_i^+$, and computing its rate using Algorithm \ref{algo:nui} on $G_i$. To determine the subgraph $G_i$
node $i$ runs the MAXCHORD algorithm of \cite{dearing1} on $G[\mathcal{N}_i^+]$ with $v_0=i$ in line $5$ (see Algorithm \ref{algo:MAXCHORD}). 
We refer to this approximation as the {\it local} chordal subgraph (LCS) approximation. When the graph $G$ is chordal this distributed
algorithm corresponds to the distributed algorithm of Section \ref{sec:dist} (as MAXCHORD determines a
{\it maximal} chordal subgraph) and is therefore exact.

Note the Bethe approximation corresponds to using the subtree consisting of the edges $(i,j)$ with $j \in \mathcal{N}_i$, instead of 
the maximal chordal subgraph $G_i$ and applying Algorithm \ref{algo:nui} on this subtree. As such the local chordal subgraph
approximation takes more conflicts of $G$ into account and can be expected to be more accurate than the Bethe
approximation. 
This was confirmed using numerical experiments. In Table \ref{tab:approx_res_locsub} we show the average relative 
deviation from the target throughout observed in a simulation run of length $10^7$ for the LCS and Bethe approximation.
We considered three different conflict graphs each with $100$ nodes with varying values for the interference threshold $R$
and the target throughput $\theta_i$, resulting in a total of $15$ setups.

\section{Beyond chordal graphs}\label{sec:beyond}

In this section we indicate that obtaining a simple closed form expression for the back-off vector
that achieves a given achievable throughput vector $\vec \theta$ appears problematic when the conflict
graph $G$ is not chordal. For this purpose we first consider the smallest non-chordal graph, that is, 
a ring network consisting of $4$ nodes (see Figure \ref{fig:ring4}). In this particular case it is still possible to obtain a closed form expression
for the back-off rates $\nu_i(\vec \theta)$, but the expression does not appear to have a very elegant form.
More importantly, this expression shows that the back-off rate of a node no longer solely depends on its
own target throughput and the target throughput of its neighbors in $G$, as demonstrated in Figure \ref{fig:rate3}.
This figure illustrates that the required back-off rate of node $3$ decreases from $3/4$ to $1/\sqrt{2}$ as $\theta_1$ increases
from $0$ to $1/4$, while the target throughput of nodes $2, 3$ and $4$ is fixed at $1/4$ and node $1$ and $3$ do
not interfere.

\begin{figure}[t]
\begin{center}
\begin{tikzpicture}[scale=0.9]
\tikzstyle{every node}=[circle, draw, fill=black!10,
                        inner sep=0pt, minimum width=10pt]

\node (1) at (1,0) []{$1$};
\node (2) at (2,0) []{$2$};
\node (3) at (3,0) []{$3$};
\node (4) at (4,0) []{$4$};

\draw [] (1) -- (2);
\draw [] (2) -- (3);
\draw [] (3) -- (4);
\draw (1) to[out=45,in=135] (4);

\end{tikzpicture}
\end{center}
\caption{Conflict graph of ring network consisting of $n=4$ nodes.}
\label{fig:ring4}
\end{figure}
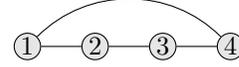

\begin{figure}[t]
\center
\includegraphics[width=0.4\textwidth]{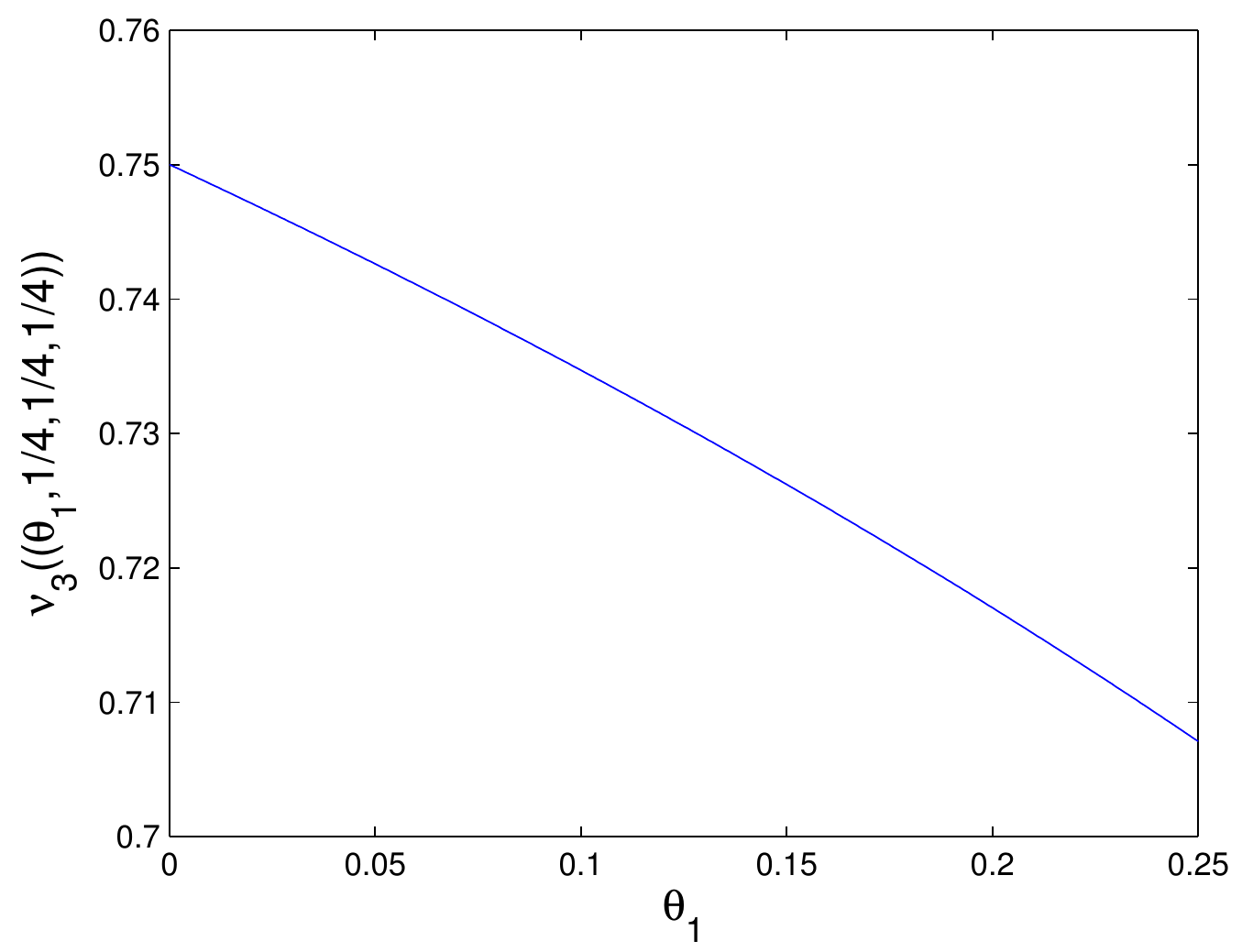}
\caption{Back-off rate of node $3$ as a function of $\theta_1$ in a ring network with $n=4$ nodes, 
where $\theta_2 = \theta_3 = \theta_4 = 1/4$. }
\label{fig:rate3}
\end{figure}

The second example consists of a ring network with $4$ nodes that is extended by a fifth node that interferes
with nodes $3$ and $4$ (see Figure \ref{fig:nonchord2}). What is worth noting here is that the neighbors
of node $5$ form a clique.  Let $\vec \theta = (\theta_1,\ldots,\theta_5)$ be an achievable target throughput 
vector and $\nu_i^{ring}((\theta_1,\ldots,\theta_4))$ the back-off rate needed to achieve the throughput vector
$(\theta_1,\ldots,\theta_4)$ in the ring network consisting of nodes $1$ to $4$ only.
One can show that $\nu_5(\vec \theta) = \theta_5/(1-(\theta_3+\theta_4+\theta_5))$,
$\nu_i(\vec \theta) =  \nu_i^{ring}((\theta_1,\ldots,\theta_4))$ for $i=1,2$ and 
\[ \nu_i(\vec \theta)  = \nu_i^{ring}((\theta_1,\ldots,\theta_4)) \frac{1-(\theta_3+\theta_4)}{1-(\theta_3+\theta_4+\theta_5)},\]
for $i=3$ and $4$, achieves the vector $\vec \theta$. In other words, when we add node $5$ to the ring network, we only need to adapt the rates
of the interfering nodes $3$ and $4$ by the same factor and set the rate of node $5$ as if it belongs to the
complete conflict graph consisting of nodes $3$, $4$ and $5$ only.

\begin{figure}[t]
\begin{center}
\begin{tikzpicture}[scale=0.9]
\tikzstyle{every node}=[circle, draw, fill=black!10,
                        inner sep=0pt, minimum width=10pt]

\node (1) at (1,0) []{$1$};
\node (2) at (2,0) []{$2$};
\node (3) at (3,0) []{$3$};
\node (4) at (4,0) []{$4$};
\node (5) at (5,0) []{$5$};

\draw [] (1) -- (2);
\draw [] (2) -- (3);
\draw [] (3) -- (4);
\draw [] (4) -- (5);
\draw (1) to[out=45,in=135] (4);
\draw (3) to[out=-45,in=-135] (5);

\end{tikzpicture}
\end{center}
\caption{Non-chordal conflict graph of a network consisting of $n=5$ nodes.}
\label{fig:nonchord2}
\end{figure}
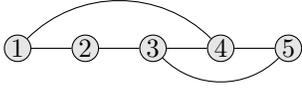

The next theorem shows that this observation holds in general as it indicates that
whenever we add a node labeled  $n+1$ to a network with $n$ nodes and
conflict graph $G$ such that the set of neighbors of $n+1$ form a clique in $G$, it suffices to adapt the back-off
rates of the neighbors of $n+1$ by a constant factor and to set the rate of node $n+1$ in the appropriate manner.
This theorem can also be used to prove the correctness of Algorithm \ref{algo:nui} directly.

\begin{theorem}
Consider a network with $n$ nodes and conflict graph $G=(V,E)$ where the nodes $\{k+1,\ldots,n\}$ form a clique
and assume the back-off rates $(\nu_1,\ldots,\nu_n)$ achieve throughput vector $(\theta_1,\ldots,\theta_n)$.
Let $\tilde G = (V \cup \{n+1\}, E \cup \cup_{i=k+1}^n (i,n+1))$ be the conflict graph of the network 
obtained by adding node $n+1$, the neighbors in $\tilde G$ of which are $\{k+1,\ldots,n\}$. 
Then, for any $\theta_{n+1} < 1-(\theta_{k+1}+\ldots +\theta_n)$, the vector $(\tilde \nu_1,\ldots,
\tilde \nu_{n+1})$ with $\tilde \nu_i= \nu_i$, for $i=1,\ldots,k$, 
\[\tilde \nu_i = \nu_i \frac{1-\sum_{j =k+1}^n \theta_j}{1-\theta_{n+1}-\sum_{j =k+1}^n \theta_j}, \]
for $i=k+1,\ldots,n$ and $\tilde \nu_{n+1} = \theta_{n+1} / (1-\theta_{n+1}-\sum_{j =k+1}^n \theta_j)$
achieves throughput vector $(\theta_1,\ldots,\theta_n,\theta_{n+1})$.
\end{theorem}
\begin{proof}
Let $\Omega_n$ be the state space of the Markov process corresponding to the network with $n$ nodes
and define $Z_{1:j}$ as the normalizing constant of the network with back-off rates
$(\nu_1,\ldots,\nu_n)$ consisting of nodes $1$ to $j$ only and $Z_n = Z_{1:n}$. Similarly define
$\tilde Z_{1:j}$ and $\tilde Z_{n+1}$ for the network consisting of $n+1$ nodes with back-off
rates $(\tilde \nu_1,\ldots,\tilde \nu_{n+1})$. As the nodes $\{k+1,\ldots, n\}$ form a clique
in $G$ we have
\[Z_n = Z_{1:k} + \sum_{i=k+1}^n \left(\sum_{\vec z \in \Omega_n} \prod_{j=1}^n \nu_j^{z_j} 1[z_i=1]\right).\]
As $\tilde \nu_i= \nu_i$ for $i=1,\ldots,k$, we have $Z_{1:k} = \tilde Z_{1:k}$ and 
\[\tilde Z_{n+1} = Z_{1:k} (1+\tilde \nu_{n+1})+ \sum_{i=k+1}^{n} \left(\sum_{\vec z \in \Omega_n} \prod_{j=1}^n \tilde \nu_j^{z_j} 1[z_i=1]\right),\]
as $\{k+1,\ldots, n+1\}$ is a clique in $\tilde G$ and node $n+1$ has no neighbors in $\{1,\ldots,k\}$. 
Further, one easily checks that $(1+\tilde \nu_{n+1}) =  (1-\sum_{j =k+1}^n \theta_j)/(1-\theta_{n+1}-\sum_{j =k+1}^n \theta_j)$
and 
\begin{align}\label{eq:prod}
\prod_{j=1}^n &\tilde \nu_j^{z_j} 1[z_i=1] = \prod_{j=1}^n \nu_j^{z_j} 1[z_i=1] \frac{1-\sum_{j =k+1}^n \theta_j}{1-\theta_{n+1}-\sum_{j =k+1}^n \theta_j}
\end{align}
for $i > k$, as the nodes $\{k+1,\ldots,n\}$ form a clique. This yields 
\begin{align}\label{eq:Zrel}
\tilde Z_{n+1} = Z_n \frac{1-\sum_{j =k+1}^n \theta_j}{1-\theta_{n+1}-\sum_{j =k+1}^n \theta_j} = Z_n (1+\tilde \nu_{n+1}).
\end{align}
Assuming the back-off rates $(\nu_1,\ldots,\nu_n)$ achieve throughput vector $(\theta_1,\ldots,\theta_n)$ in $G$,
we observe
\begin{align}\label{eq:thetai}
\theta_i = \frac{1}{Z_n} \sum_{\vec z \in \Omega_n} \prod_{j=1}^n \nu_j^{z_j} 1[z_i=1].
\end{align}
We now show that the vector $(\tilde \nu_1,\ldots,\tilde \nu_{n+1})$ achieves throughput vector $(\theta_1,\ldots,\theta_{n+1})$ in $\tilde G$.
For $i=1,\ldots,k$ and by separating the terms in which a neighbor of node $n+1$
is active or not, we can express the throughput of node $i$ in $\tilde G$ as
\begin{align*}
\frac{(1+\tilde \nu_{n+1})}{\tilde Z_{n+1}}& \left( \sum_{\vec z \in \Omega_n} \prod_{j=1}^n \nu_j^{z_j} 1[z_i=1,z_{k+1}=\ldots=z_n=0] \right)\\
&+ \frac{1}{\tilde Z_{n+1}}\left( \sum_{j=k+1}^n \sum_{\vec z \in \Omega_n} \prod_{j=1}^n \tilde \nu_j^{z_j} 1[z_i=z_j=1] \right) 
\end{align*} 
Combined with \eqref{eq:Zrel} and \eqref{eq:thetai} this implies that node $i \in \{1,\ldots,k\}$ has throughput $\theta_i$
in $\tilde G$.

For $i \in \{k+1,\ldots,n\}$ we immediately have due to \eqref{eq:prod} and \eqref{eq:Zrel} that the throughput of node $i$ can be written as
\begin{align*}
  \frac{1}{\tilde Z_{n+1}} \sum_{\vec z \in \Omega_n} \prod_{j=1}^n \tilde \nu_j^{z_j} 1[z_i=1] = 
  \frac{1}{Z_{n}} \sum_{\vec z \in \Omega_n} \prod_{j=1}^n \nu_j^{z_j} 1[z_i=1], 
\end{align*}
which equals $\theta_i$. Finally, for $i=n+1$ the throughput is given by
\[
  \frac{1}{\tilde Z_{n+1}} \sum_{\vec z \in \Omega_n} \prod_{j=1}^n \nu_j^{z_j} 1[z_{k+1}=\ldots=z_n=0] \frac{\theta_{n+1}}{1-\sum_{i=k+1}^{n+1} \theta_i} 
\]
and 
\[
  \frac{1}{Z_{n}} \sum_{\vec z \in \Omega_n} \prod_{j=1}^n \nu_j^{z_j} 1[z_{k+1}=\ldots=z_n=0] = 1-\sum_{i=k+1}^n \theta_i,
\]
as $\{k+1,\ldots,n\}$ is a clique.
By \eqref{eq:Zrel} we observe that the expression for the throughput of node $n+1$ in $\tilde G$ simplifies to $\theta_{n+1}$. 
\end{proof}

\section{Related work}\label{sec:related}
The well-known product form to analyze the throughput of idealized CSMA/CA networks with a general conflict graph
was first introduced in \cite{boorstyn1}, where its insensitivity with respect to the packet length distribution was also
shown. Insensitivity with respect to the length of the back-off period was proven much later in \cite{vandeven3,liew2}.

In \cite{durvy2} the fairness of large CSMA/CA networks was studied and for regular networks (lines and grids)
conditions on when the unfairness propagates within the network were presented. 
The cause of unfairness in CSMA/CA networks was further analyzed in \cite{durvy1}, where the equality of 
the receiving and sensing ranges was identified as an important cause.  
Fairness in CSMA/CA line networks was also studied in \cite{vandeven1}, where an explicit formula was presented
to achieve fairness in a line network consisting of $n$ nodes, where each node interferes with the next and previous $\beta$
nodes. The existence of a unique vector of back-off rates to achieve any achievable throughput vector was established
in \cite{vandeven5}, which also discusses several iterative algorithms to compute this vector. 

In \cite{jiang1,jiang2} the set of achievable throughput vectors of an ideal CSMA/CA network was identified and a dynamic
algorithm to set the back-off rates was proposed that was proven to be throughput-optimal. Generalizations of this algorithm in a
setting with packet collisions were considered in \cite{jiang3}. An simple approximate algorithm to set the back-off rates 
to achieve a given target throughput vector that requires only a single iteration was presented in \cite{yun1}.

Several generalization of the ideal CSMA/CA model have been studied. These include linear networks with
hidden and exposed nodes \cite{vandeven6}, single and multihop networks with unsaturated users \cite{laufer1}
and networks relying on multiple channels \cite{bonald2,liew1}, where it should be noted that the stability conditions
for the unsaturated network presented in \cite{laufer1} are not valid in general \cite{cecchi1}.

Another line of related work, inspired by maximum weight scheduling \cite{tassiulas1},
considers adapting the transmission lengths based on the current queue length of a node \cite{rajagopalan1,ghaderi1},
that is, backoff periods and packets have mean length $1$ and after 
a backoff period or packet transmission a node transmits a packet with some
probability that depends on a weight function of its queue length 
(provided that its neighbors are sensed silent). The main observation was
that a slowly changing weight function is necessary for stability and among such functions a slower function
leads to a more stable network at the cost of increased queue-sizes (and delay).
Other queue-length based CSMA/CA algorithms that were shown to be throughput-optimal in some setting include
\cite{marbach1,ni1}.

While throughput-optimality is very desirable, some of these queue-based algorithms have poor delay characteristics \cite{bouman1}, 
which resulted in the design of (order) delay optimal CSMA algorithms \cite{shah1,lotfinezhad1}. 
Finally, there is also a large body of work on CSMA/CA networks in the context of 802.11 networks that
was initiated by the seminal work in \cite{bianchi1}, which we do not discuss here.

\section{Conclusions}\label{sec:conclusion}
In this paper we presented closed-form expressions for the back-off vector needed to achieve
any achievable throughput vector in an ideal CSMA/CA network with a chordal conflict graph. 
These expressions are such that the back-off rate of a node only depends on its own target throughout 
and the target throughput of its neighbors. We further indicated that expressions of this type cannot
be obtained even for the simplest non-chordal graph (that is, a ring of size $4$).

We also briefly explored the possibility of defining chordal approximations for the back-off
rates for a general conflict graph. To this end we introduced a distributed approximation algorithm called
the local chordal subgraph approximation and observed that it provides more accurate results
than the Bethe approximation proposed in \cite{yun1}.
We note that the Bethe approximation is only one of many free energy approximation studied in
statistical physics. Ongoing work includes developing more complex free energy approximations for CSMA networks
and studying their relation with the results presented in this paper. 

\bibliographystyle{IEEEtran}
\bibliography{../../PhD/thesis}




\appendices

\section{Proof of Theorem~\ref{TH:TREE}}\label{apx:proof_tree}
We first establish the following lemma:

\begin{lemma}\label{lem:Ztree}
Let $I \subset V(G)$ and $E(I) = \{(k,j)\in E(G) |$ $k,j \in I\}$. Further assume $(I,E(I))$ is connected and
there is exactly one edge $(l,m) \in E(G)$ such that $l \in I$ and $m \not\in I$. Let $Z^{tree}_I(\vec \theta)$ be the normalizing constant of
the network with conflict graph $(I,E(I))$ and $Z^{tree}_n(\vec \theta)$ be the normalizing constant of the entire network  when the rates 
are set according to \eqref{eq:nui_tree}.
Then,
\begin{align}\label{eq:ZI}
Z^{tree}_I(\vec \theta) &= \frac{\prod_{j\in I} (1-\theta_{j})^{|\mathcal{N}_j|-1}}{\prod_{(k,j) \in E(I)} (1-(\theta_k+\theta_{j}))}
\frac{(1-\theta_m)}{(1-(\theta_l+\theta_m))},\\
Z^{tree}_n(\vec \theta) &= \frac{\prod_{j=1}^n (1-\theta_{j})^{|\mathcal{N}_j|-1}}{\prod_{(k,j) \in E} (1-(\theta_k+\theta_{j}))}.
\label{eq:Zn_tree}
\end{align} 
where $\mathcal{N}_i$ is the set of neighbors of $i$ in $G$.
\end{lemma}
\begin{proof}
We prove the expression for $Z^{tree}_I(\vec \theta)$ using induction on $|I|$. For $|I|=1$, we have $I = \{l\}$ 
and due to the assumption on $I$, $l$ has only one neighbor $m$ in $E$. Thus $Z^{tree}_I = (1+\nu_l^{tree}(\vec \theta)) = (1-\theta_m)/(1-(\theta_l+\theta_m))$ as
required (as the empty product is defined as $1$).

For $|I|> 1$ let $l_1,\ldots,l_{|\mathcal{N}_l|-1}$ be the neighbors of $l$ excluding $m$ and 
let $l_{s,1},\ldots,l_{s,|\mathcal{N}_{l_s}|-1}$ be the
neighbors of $l_s$ excluding $l$, for $s=1,\ldots ,|\mathcal{N}_l|-1$ (see Figure \ref{fig:tree}). Note due to the assumption on $I$ the nodes $l_s$ and $l_{s,t}$ are
all part of $I$. Define $I(l_s)$ as the connected component containing $l_s$ after removing the edge $(l,l_s)$ from $G$ and
let $I(l_{s,t})$ be the connected component containing $l_{s,t}$ after removing the edge $(l_s,l_{s,t})$ from $G$.
At this point we note that \eqref{eq:ZI} applies to the sets $I(l_s)$ and $I(l_{s,t})$ by induction as they are smaller than $I$, connected and
only one edge exists from the set. Hence, by separating the terms with $l$ inactive and active we have
\begin{align}\label{eq:Ztree1}
\MoveEqLeft Z^{tree}_I(\vec \theta) = \prod_{s=1}^{|\mathcal{N}_l|-1} Z^{tree}_{I(l_s)}(\vec \theta) + 
\nu_l^{tree}(\vec \theta) \prod_{s=1}^{|\mathcal{N}_l|-1} \prod_{t=1}^{|\mathcal{N}_{l_s}|-1} Z^{tree}_{I(l_{s,t})}(\vec \theta).
\end{align}
By induction, we have
\begin{align*}
\MoveEqLeft \prod_{s=1}^{|\mathcal{N}_l|-1} Z^{tree}_{I(l_s)}(\vec \theta) \\
&= \prod_{s=1}^{|\mathcal{N}_l|-1} \frac{\prod_{j\in I(l_s)} (1-\theta_{j})^{|\mathcal{N}_j|-1}}{\prod_{(k,j) \in E(I(l_s))} (1-(\theta_k+\theta_{j}))}
\frac{(1-\theta_l)}{(1-(\theta_l+\theta_{l_s}))}\\
&=\frac{\prod_{j\in I} (1-\theta_{j})^{|\mathcal{N}_j|-1}}{\prod_{(k,j) \in E(I)}(1-(\theta_k+\theta_{j}))}, 
\end{align*}
and similarly one finds
\begin{align*}
 \prod_{s=1}^{|\mathcal{N}_l|-1} \prod_{t=1}^{|\mathcal{N}_{l_s}|-1} &Z^{tree}_{I(l_{s,t})}(\vec \theta) \\
&=\frac{\prod_{j\in I,j\not=l} (1-\theta_{j})^{|\mathcal{N}_j|-1}}{\prod_{(k,j) \in E(I), k,j\not=l} (1-(\theta_k+\theta_{j}))}. 
\end{align*}
Therefore, by \eqref{eq:Ztree1}
\begin{align*}
Z^{tree}_I(\vec \theta) &=  \frac{\prod_{j\in I} (1-\theta_{j})^{|\mathcal{N}_j|-1}}{\prod_{(k,j) \in E(I)}(1-(\theta_k+\theta_{j}))} \\
&\hspace*{-1cm} 
+\frac{\theta_l (1-\theta_l)^{|\mathcal{N}_l|-1}}{\prod_{j \in \mathcal{N}_l} (1-(\theta_l+\theta_j))}
\frac{\prod_{j\in I,j\not=l} (1-\theta_{j})^{|\mathcal{N}_j|-1}}{\prod_{(k,j) \in E(I), k,j\not=l} (1-(\theta_k+\theta_{j}))}\\
&\hspace*{-1cm}= \frac{\prod_{j\in I} (1-\theta_{j})^{|\mathcal{N}_j|-1}}{\prod_{(k,j) \in E(I)}(1-(\theta_k+\theta_{j}))} \left(1+\frac{\theta_l}{1-(\theta_l+\theta_m)}\right).
\end{align*}
To prove the expression for $Z^{tree}_n(\vec \theta)$ assume without loss of generality that $n$ has only one neighbor, say $l$. Define $I = \{1,\ldots,n-1\}$,
let $l_1,\ldots,l_{|\mathcal{N}_l|-1}$ be the neighbors of $l$ excluding $n$ and $I(l_s)$ as the connected component containing $l_s$ 
after removing the edge $(l,l_s)$ from $G$. Then by separating the terms with node $n$ inactive and active and
relying on the expression for $Z^{tree}_I(\vec \theta)$, we have
\begin{align*}
Z^{tree}_n(\vec \theta) &= Z^{tree}_{I}(\vec \theta) + 
\nu_n^{tree}(\vec \theta) \prod_{s=1}^{|\mathcal{N}_l|-1} Z^{tree}_{I(l_s)}(\vec \theta)\\
&= \frac{\prod_{j=1}^{n-1} (1-\theta_{j})^{|\mathcal{N}_j|-1} }
{\prod_{(k,j) \in E(I)}(1-(\theta_k+\theta_{j}))}\frac{(1-\theta_n)}{(1-(\theta_l+\theta_n))}\\
&+ \frac{\theta_n}{1-(\theta_l+\theta_n)}
\frac{\prod_{j=1}^{n-1} (1-\theta_{j})^{|\mathcal{N}_j|-1}}{\prod_{(k,j) \in E(I)} (1-(\theta_k+\theta_{j}))}\\
&= \frac{\prod_{j=1}^{n} (1-\theta_{j})^{|\mathcal{N}_j|-1} }
{\prod_{(k,j) \in E}(1-(\theta_k+\theta_{j}))} ((1-\theta_n)+\theta_n),
\end{align*}
as $|\mathcal{N}_n|=1$ and $E = E(I) \cup \{(l,n)\}$.
\end{proof}

\begin{figure}[t]
\begin{center}
\begin{tikzpicture}[scale=0.9]

\draw [dashed] (-3.5,-0.5) -- (5,-0.5);

\node (a) at (-3,-1) {$I$};
\node (m) at (0,0) {$m$};
\node (l) at (0,-1) {$l$};
\node (l2) at (0.5,-2)  {$l_2$};
\node (l1) at (-2,-2) {$l_1$};
\node (ldots) at (2,-2) {$\ldots$};
\node (lf) at (4,-2) {$l_k$};

\node (l11) at (-3,-3.3) {$l_{1,1}$};
\node (l12) at (-2.4,-3.3) {$l_{1,2}$};
\node (l1dots) at (-1.8,-3.3){$\ldots$};
\node (l1f) at (-1.2,-3.3){$l_{1,c_1}$};

\node (l21) at (-0.5,-3.3) {$l_{2,1}$};
\node (l22) at (0.1,-3.3) {$l_{2,2}$};
\node (l2dots) at (0.7,-3.3){$\ldots$};
\node (l2f) at (1.3,-3.3){$l_{2,c_2}$};

\node (lf1) at (3,-3.3) {$l_{k,1}$};
\node (lf2) at (3.6,-3.3) {$l_{k,2}$};
\node (lfdots) at (4.2,-3.3){$\ldots$};
\node (lff) at (4.8,-3.3){$l_{k,c_k}$};

\draw [] (l) -- (m);
\draw [] (l) -- (l1);
\draw [] (l) -- (l2);
\draw [] (l) -- (lf);

\draw [] (l1) -- (l11);
\draw [] (l1) -- (l12);
\draw [] (l1) -- (l1f);

\draw [] (l2) -- (l21);
\draw [] (l2) -- (l22);
\draw [] (l2) -- (l2f);

\draw [] (lf) -- (lf1);
\draw [] (lf) -- (lf2);
\draw [] (lf) -- (lff);

\end{tikzpicture}
\end{center}
\caption{$k=|\mathcal{N}_l|-1$ and $c_s = |\mathcal{N}_{l_s}|-1$ for $s=1,\ldots,k$.}
\label{fig:tree}
\end{figure}
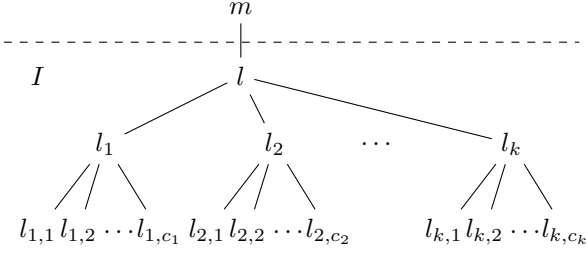

Theorem \ref{TH:TREE} can now be proven as follows. 
Let $i_1,\ldots,i_{|\mathcal{N}_i|}$ be the neighbors of $i$ and let $i_{s,1},$ $\ldots,$ $i_{s,|\mathcal{N}_{i_s}|-1}$ be the
neighbors of $i_s$ excluding $i$, for $s=1,\ldots,|\mathcal{N}_i|$. Let $I(i_{s,t})$ be the connected component containing
$i_{s,t}$ when the edge $(i_s,i_{s,t})$ is removed from $G$.
Due to the product form in \eqref{eq:stst} and Lemma \ref{lem:Ztree} the throughput of node $i$ can be written as
\begin{align*}
\frac{1}{Z^{tree}_n} & \nu_i^{tree}(\vec \theta) \prod_{s=1}^{|\mathcal{N}_i|} \prod_{t=1}^{|\mathcal{N}_{i_s}|-1} Z^{tree}_{I(i_{s,t})}(\vec \theta)\\
&= \frac{1}{Z^{tree}_n} \frac{\theta_i (1-\theta_i)^{|\mathcal{N}_i|-1}}{\prod_{j \in \mathcal{N}_i} (1-(\theta_i+\theta_j))}\\
& \hspace*{1cm} \times \frac{\prod_{j\in I,j\not=i} (1-\theta_{j})^{|\mathcal{N}_j|-1}}{\prod_{(k,j) \in E, k,j\not=i} (1-(\theta_k+\theta_{j}))}
= \theta_i. 
\end{align*} 
The uniqueness of the vector $(\nu_1^{tree}(\vec \theta),\ldots,\nu_n^{tree}(\vec \theta))$ is due to Theorem 2 in \cite{vandeven5}.

\section{Proof of Theorem~\ref{TH:LINE}}\label{apx:proof_line}

First note that the requirement $T < 1$ guarantees that $\nu_i^{line}(\vec \theta)$ is positive for $i=1,\ldots,n$.
We start with the following lemma:
\begin{lemma}\label{lem:Zline}
Let $Z_{i:j}^{line}(\vec \theta)$ be the normalizing constant for the line network consisting of nodes $i$ to $j$ only 
given that the back-off rate of node $i$ is set according to \eqref{eq:nui_line} for $i=1,\ldots,n$ and let
$Z_n^{line}(\vec \theta)=Z_{1:n}^{line}(\vec \theta)$. Then we have
for $i < n - \beta$ and $k > \beta$:
\begin{align}\label{eq:Zi}
Z_{1:i}^{line}(\vec \theta) &= \left.\prod_{j=\beta+1}^{i+\beta} g_{j-\beta+1:j}(\vec \theta)\middle/ 
\prod_{j=\beta+1}^{i+\beta} g_{j-\beta:j}(\vec \theta) \right.,\\
Z_{k:n}^{line}(\vec \theta) &= \left.\prod_{j=k-1}^{n-1} g_{j-\beta+1:j}(\vec \theta) \middle/ \prod_{j=k}^{n}g_{j-\beta:j}(\vec \theta)\right.,
\label{eq:Zk} \\
Z_n^{line}(\vec \theta) &= \left.\prod_{j=\beta+1}^{n-1}  g_{j-\beta+1:j}(\vec \theta) \middle/ \prod_{j=\beta+1}^{n} g_{j-\beta:j}(\vec \theta) \right.,\label{eq:Znorm},
\end{align} 
where $g_{j:k}(\vec \theta) = 1-(\theta_{j}+\ldots+\theta_{k})$, for $j \leq k$.
\end{lemma}
\begin{proof}
As the interference range is $\beta$, we have 
\[Z^{line}_{1:i}(\vec \theta) = 1+\sum_{j=1}^i \nu_j^{line}(\vec \theta),\] 
for $i \leq \beta+1$ and 
\[Z^{line}_{1:i}(\vec \theta) = Z^{line}_{1:i-1}(\vec \theta) + \nu_i^{line}(\vec \theta) Z^{line}_{1:i-(\beta+1)}(\vec \theta),\] 
for $i > \beta+1$. For $i=1$ we have
\begin{align*}
Z^{line}_{1:i}(\vec \theta) &= 1+\nu_1^{line}(\vec \theta) = 1+\frac{\theta_1}{g_{1:\beta+1}(\vec \theta)}
= \frac{g_{2:\beta+1}(\vec \theta)}{g_{1:\beta+1}(\vec \theta)}
\end{align*}
as required and for $2 \leq i \leq \beta+1$, we find by induction:
\begin{align*}
\MoveEqLeft Z^{line}_{1:i}(\vec \theta) = Z^{line}_{1:i-1}(\vec \theta)+\nu_i^{line}(\vec \theta)  \\
&= \frac{\prod_{j=\beta+1}^{i+\beta-1} g_{j-\beta+1:j}(\vec \theta)}{\prod_{j=\beta+1}^{i+\beta-1} g_{j-\beta:j}(\vec \theta)} + \theta_i
\frac{\prod_{j=\beta+1}^{i+\beta-1} g_{j-\beta+1:j}(\vec \theta)}{\prod_{j=\beta+1}^{i+\beta} g_{j-\beta:j}(\vec \theta)} \\
&= \frac{\prod_{j=\beta+1}^{i+\beta-1} g_{j-\beta+1:j}(\vec \theta)}{\prod_{j=\beta+1}^{i+\beta} g_{j-\beta:j}(\vec \theta)} \underbrace{\left( g_{i:i+\beta}(\vec \theta) + \theta_i \right)}_{ = g_{i+1:i+\beta}(\vec \theta)} \\
&= \frac{\prod_{j=\beta+1}^{i+\beta} g_{j-\beta+1:j}(\vec \theta)}{\prod_{j=\beta+1}^{i+\beta} g_{j-\beta:j}(\vec \theta)}.
\end{align*}
For $\beta+1 < i < n-\beta$ induction yields
\begin{align*}
Z^{line}_{1:i}(\vec \theta) &= Z^{line}_{1:i-1}(\vec \theta)+\nu_i^{line}(\vec \theta) Z^{line}_{1:i-(\beta+1)}(\vec \theta)\\  
&= \frac{\prod_{j=\beta+1}^{i+\beta-1} g_{j-\beta+1:j}(\vec \theta)}{\prod_{j=\beta+1}^{i+\beta-1} g_{j-\beta:j}(\vec \theta)} + \theta_i
\frac{\prod_{j=i}^{i+\beta-1} g_{j-\beta+1:j}(\vec \theta)}{\prod_{j=i}^{i+\beta} g_{j-\beta:j}(\vec \theta)} \\
& \hspace*{2cm} \times \frac{\prod_{j=\beta+1}^{i-1} g_{j-\beta+1:j}(\vec \theta)}{\prod_{j=\beta+1}^{i-1} g_{j-\beta:j}(\vec \theta)}\\
&= \frac{\prod_{j=\beta+1}^{i+\beta-1} g_{j-\beta+1:j}(\vec \theta)}{\prod_{j=\beta+1}^{i+\beta} g_{j-\beta:j}(\vec \theta)} \underbrace{\left( g_{i:i+\beta}(\vec \theta) + 
\theta_i \right)}_{ = g_{i+1:i+\beta}(\vec \theta)} \\
&= \frac{\prod_{j=\beta+1}^{i+\beta} g_{j-\beta+1:j}(\vec \theta)}{\prod_{j=\beta+1}^{i+\beta} g_{j-\beta:j}(\vec \theta)}.
\end{align*}
The argument to establish \eqref{eq:Zk} proceeds in a very similar fashion. 
The expression \eqref{eq:Znorm} for $Z^{line}_n(\vec \theta)$ is found by noting that at most one node can be active at any point in time within the first $\beta+1$ nodes,
that is,
\begin{align*}
\MoveEqLeft Z^{line}_n(\vec \theta) = Z^{line}_{\beta+2:n}(\vec \theta) + \sum_{i=1}^{\beta+1} \nu_i^{line}(\vec \theta) Z^{line}_{i+\beta+1:n}(\vec \theta) \\  
&= \frac{\prod_{j=\beta+1}^{n-1} g_{j-\beta+1:j}(\vec \theta)}{\prod_{j=\beta+2}^{n} g_{j-\beta:j}(\vec \theta)} + \sum_{i=1}^{\beta+1} \theta_i
\frac{\prod_{j=\beta+1}^{i+\beta-1} g_{j-\beta+1:j}(\vec \theta)}{\prod_{j=\beta+1}^{i+\beta} g_{j-\beta:j}(\vec \theta)} \\
& \hspace*{2cm} \times \frac{\prod_{j=i+\beta}^{n-1} g_{j-\beta+1:j}(\vec \theta)}{\prod_{j=i+\beta+1}^{n} g_{j-\beta:j}(\vec \theta)}\\
&= \frac{\prod_{j=\beta+1}^{n-1} g_{j-\beta+1:j}(\vec \theta)}{\prod_{j=\beta+1}^{n} g_{j-\beta:j}(\vec \theta)}\left( g_{1:\beta+1}(\vec \theta) + \sum_{i=1}^{\beta+1}
 \theta_i \right)\\
 & = \frac{\prod_{j=\beta+1}^{n-1} g_{j-\beta+1:j}(\vec \theta)}{\prod_{j=\beta+1}^{n} g_{j-\beta:j}(\vec \theta)}.
\end{align*}
\end{proof}

The proof now proceeds as follows. 
Due to Lemma \ref{lem:Zline} and the product form in \eqref{eq:stst}, the throughput of node $i$ can be written as
\begin{align*} 
\MoveEqLeft Z^{line}_{1:i-(\beta+1)}(\vec \theta) \nu_i^{line}(\vec \theta) Z^{line}_{i+\beta+1}(\vec \theta)/Z^{line}_n(\vec \theta) = \\
&= \frac{1}{Z^{line}_n(\vec \theta)}  \frac{\prod_{j=\beta+1}^{i-1} g_{j-\beta+1:j}(\vec \theta)}{\prod_{j=\beta+1}^{i-1} g_{j-\beta:j}(\vec \theta)}\\
&  \hspace*{0.5cm} \times \theta_i \frac{\prod_{j=\max(i,\beta+1)}^{\min(i+\beta,n)-1} g_{j-\beta+1:j}(\vec \theta)}{\prod_{j=\max(i,\beta+1)}^{\min(i+\beta,n)} g_{j-\beta:j}(\vec \theta)}
\frac{\prod_{j=i+\beta}^{n-1} g_{j-\beta+1:j}(\vec \theta)}{\prod_{j=i+\beta+1}^{n} g_{j-\beta:j}(\vec \theta)}\\
&= \frac{\theta_i}{Z^{line}_n(\vec \theta)} \left( \frac{\prod_{j=\beta+1}^{n-1} g_{j-\beta+1:j}(\vec \theta)}{\prod_{j=\beta+1}^{n} g_{j-\beta:j}(\vec \theta)} \right) = \theta_i.
\end{align*} 
The uniqueness of the vector $(\nu_1^{line}(\vec \theta),\ldots,\nu_n^{line}(\vec \theta))$ is due to the global
invertibility proven in \cite[Theorem 2]{vandeven5}.

\begin{IEEEbiography}[{\includegraphics[width=1in,height=1.25in,clip,keepaspectratio]{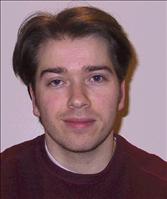}}]{Benny Van Houdt (benny.vanhoudt@uantwerpen.be) received his M.Sc. degree in Mathematics and Computer Science, and a PhD in Science from the University of Antwerp (Belgium) in July 1997, and May 2001, respectively. From August 1997 until September 2001 he held an Assistant position at the University of Antwerp. Starting from October 2001 onwards he has been a postdoctoral fellow of the FWO-Flanders. In 2007, he became a professor at the Mathematics and Computer Science Department of the University of Antwerp.
His main research interest goes to the performance evaluation and stochastic modeling of computer systems and communication networks. He has published several papers, containing both theoretical and practical contributions, in a variety of international journals (e.g., IEEE TON, IEEE JSAC, IEEE Trans. on Inf. Theory, IEEE Trans. on Comm., Performance Evaluation, Journal of Applied Probability, Stochastic Models, Queueing Systems, etc.) and in conference proceedings (e.g., ACM Sigmetrics, Networking, Globecom, Opticomm, ITC, etc.).}
\end{IEEEbiography}
\end{document}